\documentclass[11pt]{article}
\usepackage[utf8]{inputenc}
\usepackage[english]{babel}
\usepackage{hyphenat}

\usepackage{mathtools}
    \numberwithin{equation}{section}

\usepackage{amsfonts}
\usepackage{amssymb}
\usepackage{amsthm}
\usepackage{stmaryrd}
\usepackage{hyperref}
\usepackage{tikz-cd}
\usepackage{tikz}
\usepackage{fullpage}
\usepackage{fancyhdr}
\usetikzlibrary{arrows.meta,positioning,calc}

\newtheorem{theorem}{Theorem}[section]
\newtheorem{lemma}[theorem]{Lemma}
\newtheorem{prop}[theorem]{Proposition}
\newtheorem{corollary}[theorem]{Corollary}

\theoremstyle{definition}
\newtheorem{definition}[theorem]{Definition}

\theoremstyle{remark}
\newtheorem*{remark}{Remark}

\theoremstyle{remark}

\newcommand{\End}{\operatorname{End}}

\newcommand{\R}{\mathbb{R}}
\newcommand{\C}{\mathbb{C}}
\newcommand{\F}{\mathbb{F}}

\title{Quantized Quiver Varieties and the \\ Quantum Spin Ruijsenaars--Schneider Model}
\author{Gleb Arutyunov\footnote{\href{mailto://gleb.arutyunov@desy.de}{\texttt{gleb.arutyunov@desy.de}} (\href{https://orcid.org/0009-0009-8862-6959}{orcid.org/0009-0009-8862-6959}), \href{mailto://lukas.hardi@desy.de}{\texttt{lukas.hardi@desy.de}} (\href{https://orcid.org/0009-0005-0592-8486}{orcid.org/0009-0005-0592-8486}), II. Institut für Theoretische Physik, Universität Hamburg, Luruper Chaussee 149, 22761 Hamburg, Germany} \and Lukas Hardi${}^*$}

\begin{document}

\maketitle

\thispagestyle{fancy}
\rhead{hep-th/2508.07862, ZMP-HH/25-15}

\begin{abstract}
\noindent This paper tackles the long-standing problem of quantizing the rational spin Ruijsenaars--Schneider model originating in the work of Krichever and Zabrodin \cite{krichever:1995}. We make use of the technique of quantum Hamiltonian reduction to construct a quantized quiver variety $\mathfrak{A}_{N,\ell}$, which is simultaneously the algebra of quantum observables of the rational spin Ruijsenaars--Schneider model of $N$ particles with $\ell$ spin polarizations. Inside this algebra, we find a loop algebra and Yangian of $\mathfrak{gl}_\ell$ and conjecture that the algebra $\mathfrak{A}_{N,\ell}$ can be identified with a truncated Yangian of affine type $A_{\ell-1}^{(1)}$. Finally, we use the commutation relations inside $\mathfrak{A}_{N,\ell}$ to derive a difference equation for eigenstates of the lowest Hamiltonian that reproduces the known quantization of the spinless case when $\ell=1$.
\end{abstract}

\tableofcontents

\section{Introduction}

Integrability is a key phenomenon in models of physics that allows for exact solutions. Among integrable models is a wide class of many-body systems, one of the most general of which is the Ruijsenaars--Schneider model with elliptic potential \cite{ruijsenaars:1987}, which appears ubiquitously in modern applications. In recent years, it has become increasingly apparent that Ruijsenaars--Schneider models with internal degrees of freedom (`spins') are interesting to consider, among others appearing in recent work on long-range spin chains via the freezing method \cite{lamers:2022,klabbers:2024} and higher rank generalizations of extended quantum groups such as quantum toroidal algebras appearing in the context of Alday--Gaiotto--Tachikawa duality and instanton moduli spaces \cite{matsuo:2024,zenkevich:2024}.

The classical equations of motion of the spin generalization of Ruijsenaars--Schneider models were obtained by Krichever and Zabrodin in \cite{krichever:1995} as certain solutions of the two-dimensional Toda chain. Their work was concerned with the case of the most general elliptic potential. In its wake, much work has been done to obtain the Hamiltonian and Poisson structure underlying the equations of motion \cite{arutyunov:1998,soloviev:2008,chalykh:2020,arutyunov:2019,fairon:2021}. This work relied on the method of constructing integrable models using Hamiltonian reduction, developed in \cite{gorsky:1995}, among others. In this paper, we apply the technique of quantum Hamiltonian reduction, starting from the knowledge of the classical Hamiltonian reduction of the rational spin Ruijsenaars--Schneider model with $N$ particles and $\ell$ spin polarizations elucidated in \cite{arutyunov:1998} and re-framed in the language of quiver varieties. The other key ingredient is a presentation of the quantum cotangent bundle of $\mathrm{GL}_N$ constructed in \cite{arutyunov:1996} from a system of three $R$-matrices. By this method, we construct the algebra $\mathfrak{A}_{N,\ell}$ of quantum observables of the rational spin Ruijsenaars--Schneider model as a quantized quiver variety associated to the chainsaw quiver with $\ell$ nodes of dimension $N$.

Let us take a moment to situate our model within the landscape of integrable models. The non-relativistic limit of the rational spin Ruijsenaars--Schneider model is the rational spin Calogero--Moser--Sutherland model, which was recently considered as a model of the non-abelian quantum Hall effect in terms of the dynamics of vortices in Chern--Simons theory \cite{bourgine:2024,hu:2024}, where it was related to the deformed double current algebra, which may be seen as the rational limit of the affine Yangian $Y_\hbar(\widehat{\mathfrak{gl}}_\ell)$ \cite{guay:2005,costello:2017,gaiotto:2024}. The appearance of large quantum groups is a common property of the class of superintegrable systems, to which the rational spin Ruijsenaars--Schneider belongs \cite{reshetikhin:2016}. The quantum elliptic \emph{spinless} Ruijsenaars--Schneider model has also been related to the $\ell=1$ quantum toroidal algebra $U_{q_1,q_2}(\ddot{\mathfrak{gl}}_1)$ (a multiplicative version of the affine Yangian $Y_\hbar(\widehat{\mathfrak{gl}}_1)$) in the context of brane configurations in type IIB string theory in \cite{zenkevich:2024}. The rational spin Ruijsenaars--Schneider model is related to the quantum trigonometric spin Calogero--Moser--Suther\-land model by bispectral duality \cite{ruijsenaars:1988}. The quantum trigonometric \emph{spinless} Calogero--Moser--Sutherland model was considered in the language of quantized quiver varieties in \cite{wen:2024,varagnolo:2010} and in the language of spherical functions in \cite{etingof:1994}. The $N \to \infty$ limit of the quantum trigonometric spin Calogero--Moser--Sutherland model has been considered in \cite{uglov:1996,khoroshkin:2017,nazarov:2019b} and explicitly related to the Fock space representation of the affine Yangian $Y_\hbar(\widehat{\mathfrak{gl}}_\ell)$ in \cite{kodera:2016}. We conjecture that our algebra $\mathfrak{A}_{N,\ell}$ can be identified with a truncated affine Yangian.

It is known \cite{guay:2005} that the affine Yangian is generated by the horizontal affine and vertical Yangian subalgebras. Similarly, we find generators $\mathbf{J}^{\alpha\beta}[n]$ and $\mathbf{T}^{\alpha\beta}(z)$ inside the algebra $\mathfrak{A}_{N,\ell}$ that satisfy the Lie algebra relations of the loop algebra and the RTT relation of the Yangian. We also introduce novel generators $\mathbf{S}^{\alpha\beta}(z)$, which we call the \emph{quantum current}, from which the positive modes $\mathbf{J}^{\alpha\beta}[n]$, $n>0$, of the loop algebra can be built up. The quantum current satisfies a type of RSRS relation that first appeared in \cite{reshetikhin:1990}. We furthermore write down the commutation relation between $\mathbf{T}(z)$ and $\mathbf{S}(z)$, which together with the RTT and RSRS relations gives an FRT-like presentation of half our algebra. It is our hope that these relations can be used to give a simple FRT-like presentation of affine Yangians.

\vspace{0.3\baselineskip}

\noindent\textbf{Quiver variety/gauge theory perspective.} In \cite{arutyunov:1998}, the Poisson structure of the rational spin Ruijsenaars--Schneider model was obtained from the Hamiltonian reduction
\begin{equation}
\mathfrak{M}_{N,\ell} \coloneq (T^* \mathrm{GL}_N \times T^* \C^{N \times \ell}) \sslash_\gamma \mathrm{GL}_N,
\end{equation}
where $T^* \mathrm{GL}_N$ and $T^* \C^{N \times \ell}$ come equipped with the canonical Poisson structures and the coupling constant $\gamma \in \C$ parametrizes a choice of coadjoint orbit proportional to the identity matrix. The present paper constructs the algebra $\mathfrak{A}_{N,\ell}$ of quantum observables of the rational spin Ruijsenaars--Schneider model via the corresponding \emph{quantum} Hamiltonian reduction, quantizing the Poisson algebra of functions on $\mathfrak{M}_{N,\ell}$. Notice that $\mathfrak{M}_{N,\ell}$ is \emph{not} the Nakajima quiver variety for the framed Jordan quiver, since we consider $T^* \mathrm{GL}_N$ instead of $T^* \mathfrak{gl}_N$. Rather, $\mathfrak{M}_{N,\ell}$ is a Cherkis bow variety, specifically the chainsaw quiver variety considered in \cite{nakajima:2017} associated to chainsaw quiver representations of the type
\begin{center}
    \begin{tikzcd}[column sep=0.5cm]
        \cdots \arrow{rr}{g_{\ell-1}} \arrow[swap]{dr}{b^{\ell-1}} & & \C^N \arrow{rr}{g_\ell} \arrow[swap]{dr}{b^\ell} \arrow[loop left,in=60,out=120,looseness=5]{}{x_{\ell-1}} & & \C^N \arrow{rr}{g_1} \arrow[swap]{dr}{b^1} \arrow[loop left,in=60,out=120,looseness=5]{}{x_\ell} & & \C^N \arrow{rr}{g_2} \arrow[swap]{dr}{b^2} \arrow[loop left,in=60,out=120,looseness=5]{}{x_1} & & \cdots \arrow[bend right,in=-6,out=186,opacity=0.2]{llllllll} \\
        & \C \arrow[swap]{ur}{a^{\ell-1}} & & \C \arrow[swap]{ur}{a^\ell} & & \C \arrow[swap]{ur}{a^1} & & \C \arrow[swap]{ur}{a^2} &
    \end{tikzcd}
\end{center}
To clarify the isomorphism of $\mathfrak{M}_{N,\ell}$ with this chainsaw quiver variety, we remark that we can use the stability conditions and the moment map equation of the chainsaw quiver variety to show that the relevant data of the chainsaw quiver representations are entirely summed up by the tuple
\begin{equation}
    (g_\ell \cdots g_1, x_\ell,(a^\alpha)_{\alpha=1}^\ell,(b^\alpha)_{\alpha=1}^\ell) \in \mathrm{GL}_N \times \C^{N \times N} \times \C^{N \times \ell} \times \C^{\ell \times N},
\end{equation}
which specifies an element of $T^* \mathrm{GL}_N \times T^* \C^{N \times \ell}$\footnote{See \cite[\S 3.1.1]{nakajima:2017} for a more detailed discussion of the Poisson structure.}. We note that by the results of \cite{nogradi:2005,takayama:2016}, $\mathfrak{M}_{N,\ell}$ can be identified with the moduli space of rank $\ell$ \emph{periodic} instantons living on the spacetime $\R^3 \times S^1$, also known as \emph{calorons}, with the extra conditions of zero total magnetic charge and that the Polyakov loop has distinct eigenvalues. By \cite{nakajima:2017}, $\mathfrak{M}_{N,\ell}$ can also be identified with the Coulomb branch of the quiver gauge theory for the $A_{\ell-1}^{(1)}$-type quiver with dimension vector $(N,\dots,N)$:
\begin{center}
    \def\l{6}
    \def\radius{1.3cm}
    \begin{tikzpicture}[>=Stealth, every node/.style={font=\small}]
        \foreach \i in {1,...,\l} {
            \pgfmathsetmacro{\angle}{360*(\i-1)/\l}
            \node[circle,draw,inner sep=1.5pt] (v\i) at (\angle:\radius) {$N$};
        }
        \node[right=2pt of v1] (d1) {\small $1$};
        \node[above=2pt of v2] (d2) {\small $2$};
        \node[above=2pt of v3] (d3) {\small $3$};
        \node[left=2pt of v4] (d4) {\small $\vdots$};
        \node[below=2pt of v5] (d5) {\small $\ell-1$};
        \node[below=2pt of v6] (d6) {\small $\ell$};
        \foreach \i in {1,...,\l} {
            \pgfmathtruncatemacro{\next}{int(mod(\i,\l)+1)}
            \draw[->, shorten >=2pt, shorten <=2pt] (v\i) to[bend right=14] (v\next);
        }
        \node at (0,0) [font=\small,fill=white,inner sep=2pt] {$A_{\ell-1}^{(1)}$};
    \end{tikzpicture}
\end{center}
It was conjectured in \emph{loc. cit.} that the quantum Hamiltonian reduction procedure for $\mathfrak{M}_{N,\ell}$ produces the quantized Coulomb branch of this quiver gauge theory. If true, we expect the algebra $\mathfrak{A}_{N,\ell}$ of quantum observables of the rational spin Ruijsenaars--Schneider model to be identified with the truncated shifted affine Yangian $Y_\mu^\lambda(\widehat{\mathfrak{sl}}_\ell)$ with truncation coweight $\lambda = N \delta$ and shift coweight $\mu = 0$ (no shift), where $\delta = \sum_{i=0}^{\ell-1} \alpha_i$ is the sum of simple roots\footnote{See \cite[\S 7.3]{nakajima:2017} for a discussion of how $\lambda$ and $\mu$ depend on the dimension vector of the quiver.}.

\vspace{0.3\baselineskip}

\noindent\textbf{Layout.}
The rest of the paper is divided as follows. Section \ref{section:quantumCotangentBundle} introduces the quantum cotangent bundle of $\mathrm{GL}_N$, denoted by $\mathcal{O}_\hbar(T^*\mathrm{GL}_N)$. In \S2.1 we recast the algebra $\mathcal{O}_\hbar(T^*\mathrm{GL}_N)$ in terms of functional operators via the Lagrange basis, explicitly exhibiting the $R$‐matrices in this basis. Section \ref{section:quantumHamiltonianReduction} then carries out the quantum Hamiltonian reduction. In \S3.1 we identify the quantum moment map and compute the induced adjoint action on the unreduced algebra of observables. In \S3.2 we determine generators of the subalgebra invariant under the mirabolic subgroup. In \S3.3 we discuss methods of fixing the residual torus gauge: (i) by adjoining inverses and forming torus‐invariant combinations, and (ii) via the Dirac method. Section \ref{section:reducedAlgebra} analyzes the reduced algebra $\mathfrak{A}_{N,\ell}$. In \S4.1 we define a generating function $\mathbf{S}[n](z)$ of the generators of $\mathfrak{A}_{N,\ell}$, derive a contour‐integral recurrence that expresses $\mathbf{S}[n](z)$ for positive $n$ in terms of $\mathbf{S}[1](z)$, and prove that the residues of $\mathbf{S}[1](z)$ satisfy the RSRS-type relations. In \S4.2 we exhibit two important subalgebras inside $\mathfrak{A}_{N,\ell}$: A loop algebra $L(\mathfrak{gl}_\ell)$ whose infinite-dimensional center furnishes the commuting Hamiltonians, and a Yangian $Y(\mathfrak{gl}_\ell)$ that acts via raising operators. Section \ref{section:spectrum} presents an explicit calculation of the energy spectrum and eigenstates of the lowest quantized spin Ruijsenaars–Schneider Hamiltonian by leveraging the Yangian $Y(\mathfrak{gl}_\ell)$ as raising operators. Technical background on quantum Hamiltonian reduction is collected in appendix \ref{appendix:reduction}, and our notations are summarized on page \pageref{section:notation}.

\section{Quantum cotangent bundle of $\mathrm{GL}_N$} \label{section:quantumCotangentBundle}

We start this section with a brief review of \cite{arutyunov:1998}, where the Poisson structure of the rational spin Ruijsenaars--Schneider model was obtained via Hamiltonian reduction with the unreduced phase space given by the cotangent bundle $T^* \mathrm{GL}_N \times T^* \C^{N \times \ell}$. This is parametrized by a tuple of matrices $(g,x,a,b) \in \mathrm{GL}_N \times \mathfrak{gl}_N \times \C^{N \times \ell} \times \C^{\ell \times N}$ with coefficients $g_{ij},x_{ij},a_i^\alpha,b_i^\alpha$ for $i,j=1,\dots,N$ and $\alpha=1,\dots,\ell$ equipped with the Poisson structure
\begin{align}
\{ x_1,x_2 \} = C_{12} x_1-x_1 C_{12}, \quad \{ x_1,g_2 \} = g_2 C_{12}, \quad \{ a_1^\alpha,b_2^\beta \} = - \delta^{\alpha\beta} \Delta_{12},
\end{align}
where $C_{12} = \sum_{i,j=1}^N e_{ij} \otimes e_{ji} \in \End(\C^N)^{\otimes 2}$ and $\Delta_{12} = \sum_{i=1}^N e_i \otimes e_i^t \in \C^N \otimes \C^{N*}$, and all other pairs Poisson-commuting. Here, $e_{ij}$ are the standard matrix units and $e_i$ are the standard unit vectors. The reduction is performed with respect to the $\mathrm{GL}_N$-action
\begin{equation}
h \cdot (g,x,a,b) \coloneq (hgh^{-1},hxh^{-1},ha,bh^{-1})
\end{equation}
and the moment map equation amounts to
\begin{equation}
x - gxg^{-1} - ab = -\gamma,
\end{equation}
where $\gamma \in \C$ later becomes the coupling constant of the model.

To perform the reduction, one restricts to the open dense locus where the matrix $x$ has distinct eigenvalues, allowing for a unique diagonalization $x = TQT^{-1}$, where $Q = \operatorname{diag}(q_1,\dots,q_N)$ with ordered and pairwise distinct $q_1,\dots,q_N$, and $T$ satisfies the mirabolic condition $Te = e$, where
\begin{equation}
    e = (1,\dots,1)^t
\end{equation}
is the column vector only containing ones. While $Q$ uniquely labels orbits, $T$ gives a parameterization of each orbit and satisfies the Kostant--Kirillov bracket. Similarly, one diagonalizes $gxg^{-1} = UQU^{-1}$ with $Ue = e$. It follows that $g = UPT^{-1}$ for some diagonal matrix $P = \operatorname{diag}(P_1,\dots,P_N)$. The transformation properties of these variables under an element $h \in \mathrm{GL}_N$ are as follows:
\begin{align}
h \cdot T = hTh[T], \quad h \cdot U = hUh[U], \quad h \cdot P = h[U]^{-1} P h[T],
\end{align}
where $h[T]$ is the unique diagonal matrix satisfying $h[T] e = T^{-1} h^{-1} e$ and similarly $h[U]$ is the unique diagonal matrix satisfying $h[U] e = U^{-1} h^{-1} e$. We note that $h[T] = 1$ and $h[U] = 1$ whenever $h$ satisfies the mirabolic condition $he = e$, i.e. $h$ lies in the \emph{mirabolic subgroup}
\begin{equation}
\mathrm{M}_N \coloneq \{ h \in \mathrm{GL}_N \mid he = e \}.
\end{equation}
Its Lie algebra is the \emph{mirabolic subalgebra}
\begin{equation}
\mathfrak{m}_N \coloneq \{ h \in \mathfrak{gl}_N \mid he = 0 \}.
\end{equation}
Introducing $t$ and $u$ to be the unique diagonal matrices satisfying $te = \sum_{\rho=1}^\ell T^{-1} a^\rho$ and $ue = \sum_{\rho=1}^\ell U^{-1} a^\rho$, they transform as
\begin{equation}
h \cdot t = h[T]^{-1} t, \quad h \cdot u = h[U]^{-1} u.
\end{equation}
Combining the variables $T,U,Q,P,t,u,a,b$ appropriately yields a complete set of invariants given by the coefficients $q_1,\dots,q_N$ of $Q$ as well as
\begin{align}
\mathbf{a} \coloneq t^{-1} T^{-1} a, \quad \mathbf{c} \coloneq b UPt, \quad \mathbf{L} \coloneq t^{-1} T^{-1} UPt, \quad \mathbf{P} \coloneq u^{-1} P t.
\end{align}
In terms of these variables, the fixing of the moment map reads
\begin{align}
Q\mathbf{L}-\mathbf{L}Q + \gamma \mathbf{L} = \mathbf{a} \mathbf{c} \qquad \Leftrightarrow \qquad \mathbf{L}_{ij} = \frac{\mathbf{a}_i \mathbf{c}_j}{q_{ij}+\gamma}.
\end{align}
The fully reduced Poisson algebra is thus generated by the variables $Q,\mathbf{P},\mathbf{a}$, and $\mathbf{c}$, while $\mathbf{L}$ can be expressed in terms of these generators, and the Poisson bracket is induced from the Poisson bracket on the cotangent bundle.

Moving towards the quantization, let us describe the \emph{quantum cotangent bundle of $\mathrm{GL}_N$}, which is the algebra that quantizes the Poisson structure on the cotangent bundle $T^*\mathrm{GL}_N$ in the parameterization via $T,U,Q,$ and $P$ discussed above. It was introduced in \cite{arutyunov:1996} to derive the commutation relations satisfied by the quantum Lax operator of the rational \emph{spinless} ($\ell=1$) Ruijsenaars--Schneider model. In the following, the notation for the variables will be the same as in the classical case, but should be interpreted as quantum operators without further comment.

\begin{definition}
Let $\hbar \in \C^\times$. The \emph{quantum cotangent bundle of $\mathrm{GL}_N$}, denoted $\mathcal{O}_\hbar(T^* \mathrm{GL}_N)$, is the  unital associative algebra generated by the coefficients of matrices $T^{\pm 1},U^{\pm 1},Q,P^{\pm 1}$ with $T^{-1} T = TT^{-1} = 1, U^{-1} U = UU^{-1} = 1, P^{-1} P = PP^{-1} = 1$ and $Q,P^{\pm 1}$ diagonal, subject to the commutation relations
\begin{equation}\label{eq:qCotBundle}
\begin{aligned}
T_1 T_2 &= T_2 T_1 R_{12}, \\
T_1 P_2 &= P_2 T_1 \bar R_{12},
\end{aligned}
\qquad
\begin{aligned}
U_1 U_2 &= U_2 U_1 R_{21}, \\
U_1 P_2 &= P_2 U_1 \bar R_{12},
\end{aligned}
\qquad
\begin{aligned}
P_2 Q_1 &= (Q_1 - \hbar Y_{12}) P_2,
\end{aligned}
\end{equation}
and all other pairs commuting. Here, we have introduced
\begin{equation}
R = 1 + \sum_{i \neq j} \frac{\hbar}{q_{ij}} f_{ij} \otimes f_{ji}, \quad \bar R = 1 + \sum_{i \neq j} \frac{\hbar}{q_{ij}-\hbar} f_{ij} \otimes e_{jj}, \quad
Y \coloneq \sum_{i=1}^N e_{ii} \otimes e_{ii},
\end{equation}
where $f_{ij} \coloneq e_{ii}-e_{ij}$ is a basis for the mirabolic subalgebra $\mathfrak{m}_N$, and subscripts label auxiliary spaces. It is important to note that $R$ satisfies the Yang-Baxter equation and is unitary:
\begin{align}
R_{12} R_{13} R_{23} = R_{23} R_{13} R_{12}, \quad R^{-1} = R_{21},
\end{align}
ensuring the consistency of the commutation relations (\ref{eq:qCotBundle}). Further, it was shown in \cite{arutyunov:1996} that the combinations $x \coloneq TQT^{-1}$ and $g \coloneq UPT^{-1}$ satisfy the commutation relations
\begin{align}
[x_1,x_2] = \hbar(C_{12} x_1-x_1 C_{12}), \quad [x_1,g_2] = \hbar g_2 C_{12}, \quad [g_1,g_2] = 0,
\end{align}
which is the canonical quantization of the canonical Poisson structure on $T^* \mathrm{GL}_N$.
\end{definition}

\begin{remark}
Let us note that the mirabolic subgroup originally appeared under the name \emph{Frobenius subgroup} in \cite{arutyunov:1996}. The name \emph{mirabolic subgroup} came to our attention through the work \cite{kalmykov:2025}. The $R$-matrix recovered in \emph{loc. cit.} are closely related to the $R$-matrix $R$ in this paper. The more recent work \cite{kalmykov:2025b} building on \cite{kalmykov:2025} contains a quantum Hamiltonian reduction similar to ours that results in the quantum Toda lattice, whose classical limit is known \cite{feher:1997} to arise from the classical Hamiltonian reduction $N_+ \bbslash T^* \mathrm{GL}_N \sslash N_-$ where $N_\pm$ are the subgroups of upper and lower unipotent matrices. The quantum Hamiltonian reduction is furthermore identified with a quotient of the $(-2N)$-shifted Yangian of $\mathfrak{sl}_2$.
\end{remark}

\subsection{Quantum cotangent bundle of $\mathrm{GL}_N$ via functional operators}

A useful conceptual shift for this work is the reinterpretation of the presentation of the quantum cotangent bundle of $\mathrm{GL}_N$ in terms of functional operators. To this end, let $\F \coloneq \C(q_1,\dots,q_N)$ be the field of rational functions in the variables $q_1,\dots,q_N$. Consider the space $\F[z]_N$ of polynomials in $z$ with coefficients in $\F$ and degree strictly less than $N$. Such polynomials are uniquely determined by their values on $q_1,\dots,q_N$. In this vein, a convenient basis of $\F[z]_N$ for our purposes is given by the Lagrange interpolation polynomials
\begin{equation}
\ell_i(z) = \prod_{j(\neq i)} \frac{z-q_j}{q_i-q_j}, \quad i=1,\dots,N,
\end{equation}
which satisfy the Lagrange interpolation formula:
\begin{equation}
f(z) = \sum_{i=1}^N f(q_i) \ell_i(z), \quad \forall f(z) \in \F[z]_N.
\end{equation}
Consequently, in this basis, the special vector $e = (1,\dots,1)$ defining the mirabolic condition is nothing but the constant polynomial $1$:
\begin{equation}
    1 = \sum_{i=1}^N \ell_i(z).
\end{equation}
Let us set $\chi(z) \coloneq \det(z-Q)$ and introduce the inner product
\begin{equation}
\langle -,- \rangle \colon \F[z]_N \otimes_\F \F[z]_N \to \F, \quad f(z) \otimes g(z) \mapsto \oint \frac{dz}{\chi(z)} f(z) g(z),
\end{equation}
where $\oint dz$ takes the negative residue at infinity, i.e. the sum of all residues away from infinity. This inner product is evidently bilinear and symmetric as well as non-degenerate with the dual basis to $\ell_i(z)$ given by
\begin{equation}
\ell_i^*(z) \coloneq \prod_{j(\neq i)} (z-q_j),
\end{equation}
due to the identity
\begin{align}
\langle \ell_i^*(z), \ell_j(z) \rangle = \oint \frac{dz}{z-q_i} \ell_j(z) = \ell_j(q_i) = \delta_{ij}.
\end{align}
With this in hand, we may express $R_{12}$ as a functional operator on $\F[z_1]_N \otimes_\F \F[z_2]_N$. Let
\begin{align}
\mathbf{R}_{12}^* f(z_1,z_2) = f(z_1,z_2) + \hbar \partial_{12} f(z_1,z_2),
\end{align}
where we have introduced Newton's divided difference operator
\begin{equation}
\partial_{12} f(z_1,z_2) \coloneq \frac{f(z_1,z_2)-f(z_2,z_1)}{z_1-z_2},
\end{equation}
which preserves the space $\F[z_1]_N \otimes_\F \F[z_2]_N$. The operator $\mathbf{R}_{12}^*$ satisfies the Yang-Baxter equation when acting on $\F[z_1]_N \otimes_\F \F[z_2]_N \otimes_\F \F[z_3]_N$ and it is quickly checked that
\begin{align}
\langle \mathbf{R}_{12}^* \ell_i^*(z_1) \ell_j^*(z_2), \ell_k(z_1) \ell_l(z_2) \rangle = \langle R_{12}^t e_i \otimes e_j, e_k \otimes e_l \rangle,
\end{align}
where on the right hand side we use the standard inner product on $\C^N \otimes \C^N$. This parallels the construction of \cite{felder:1994}, where $R$-matrices are obtained as the restriction of functional operators to a finite-dimensional subspace. The adjoint of $\partial_{12}$ with respect to the inner product defined above is given by
\begin{equation}
\partial_{12}^* f(z_1,z_2) \coloneq \frac{f(z_1,z_2)+f(z_2,z_1)-f(z_1,z_1)-f(z_2,z_2)}{z_1-z_2},
\end{equation}
The operator $\partial_{12}^*$ again preserves the space $\F[z_1]_N \otimes_\F \F[z_2]_N$. In fact, $\partial_{12}^* f(z_1,z_2)$ always lies in the ideal generated by $z_1-z_2$ and fulfills $(z_1-z_2)^{-1} \partial_{12}^* (z_1-z_2) = \partial_{12}$. Letting $\mathbf{R}_{12} \coloneq 1 + \hbar\partial_{12}^*$ we hence obtain
\begin{align}
\langle \ell_i^*(z_1) \ell_j^*(z_2), \mathbf{R}_{12} \ell_k(z_1) \ell_l(z_2) \rangle = \langle e_i \otimes e_j, R_{12} e_k \otimes e_l \rangle,
\end{align}
giving the square of equalities
\begin{equation}
\begin{tikzcd}
\langle \ell_i^*(z_1) \ell_j^*(z_2), \mathbf{R}_{12} \ell_k(z_1) \ell_l(z_2) \rangle \arrow[equal]{r} \arrow[equal]{d} & \langle e_i \otimes e_j, R_{12} e_k \otimes e_l \rangle \arrow[equal]{d} \\
\langle \mathbf{R}_{12}^* \ell_i^*(z_1) \ell_j^*(z_2), \ell_k(z_1) \ell_l(z_2) \rangle \arrow[equal]{r} & \langle R_{12}^t e_i \otimes e_j, e_k \otimes e_l \rangle.
\end{tikzcd}
\end{equation}

\begin{definition}
Let us write $\ell^*(z) \coloneq \chi(z) (z-Q)^{-1}$ and $\ell(z) = \chi'(Q)^{-1} \ell^*(z)$. These are the diagonal matrices whose diagonal entries are $\ell_i^*(z)$ and $\ell_i(z)$. Define the combinations
\begin{equation}
\begin{aligned}
T(z) &\coloneq T \ell^*(z) e, \quad T^{-1}(z) \coloneq e^t \ell(z) T^{-1}, \\
U(z) &\coloneq U \ell^*(z) e, \quad U^{-1}(z) \coloneq e^t \ell(z) U^{-1}, \\
P(z) &\coloneq P \ell^*(z) e, \quad P^{-1}(z) \coloneq e^t \ell(z) P^{-1}.
\end{aligned}
\end{equation}
Then $T(z),U(z)$ and $T^{-1}(z),U^{-1}(z)$ transform as vectors and covectors under $\mathrm{M}_N$, respectively, and both depend polynomially on $z$.
\end{definition}

\begin{remark}
The combinations $T(z),U(z),P(z)$ are invertible in the following sense:
\begin{equation}
\begin{aligned}
T^{-1}(z) T(z') &= \delta(z,z'), \\
U^{-1}(z) U(z') &= \delta(z,z'), \\
P^{-1}(z) P(z') &= \delta(z,z'),
\end{aligned}
\qquad
\begin{aligned}
\langle T(z), T^{-1}(z) \rangle &= 1, \\
\langle U(z), U^{-1}(z) \rangle &= 1, \\
\langle P(z), P^{-1}(z) \rangle &= 1.
\end{aligned}
\end{equation}
Here, the left equalities are equalities of polynomials and the right equalities are equalities of $N \times N$ matrices and we make use of the $\delta$-polynomial
\begin{equation}
\delta(z,z') \coloneq \sum_{i=1}^N \ell_i(z) \ell_i^*(z').
\end{equation}
\end{remark}

\begin{prop}
The relations of $\mathcal{O}_\hbar(T^* \mathrm{GL}_N)$ may equivalently be written as $P_2 Q_1 = (Q_1 - \hbar Y_{12}) P_2$ together with
\begin{equation}
\begin{aligned}
T_1(z_1) T_2(z_2) &= \mathbf{R}_{12}^* T_2(z_2) T_1(z_1), \\
T_1^{-1}(z_1) T_2^{-1}(z_2) &= \mathbf{R}_{12} T_2^{-1}(z_2) T_1^{-1}(z_1), \\
P_1^{-1}(z_1) T_2(z_2) &= T_2(z_2) P_1^{-1}(z_1), \\
P_1^{-1}(z_1) T_2^{-1}(z_2) &= \mathbf{R}_{21}^* [T_2^{-1}(z_2)] P_1^{-1}(z_1),
\end{aligned}
\qquad
\begin{aligned}
U_1(z_1) U_2(z_2) &= \mathbf{R}_{21}^* U_2(z_2) U_1(z_1), \\
U_1^{-1}(z_1) U_2^{-1}(z_2) &= \mathbf{R}_{21} U_2^{-1}(z_2) U_1^{-1}(z_1), \\
P_1^{-1}(z_1) U_2(z_2) &= U_2(z_2) P_1^{-1}(z_1), \\
P_1^{-1}(z_1) U_2^{-1}(z_2) &= \mathbf{R}_{21}^* [U_2^{-1}(z_2)] P_1^{-1}(z_1).
\end{aligned}
\end{equation}
\end{prop}

\begin{proof}
The first and second lines are derived from the fact that $\mathbf{R}_{12}^*$ and $\mathbf{R}_{12}$ are represented by $R_{12}^t$ and $R_{12}$ in the Lagrange basis. For the third line, we have the identity $P_1^{-1} \ell_2^*(z_2) e_2 = \bar R_{21}^{-1} \ell_2^*(z_2) e_2 P_1^{-1}$, which then implies
\begin{align*}
P_1^{-1} T_2 \ell_2^*(z_2) e_2 &= T_2 \bar R_{21} P_1^{-1} \ell_2^*(z_2) e_2 = T_2 \ell_2^*(z_2) e_2 P_1^{-1} \\
P_1^{-1} U_2 \ell_2^*(z_2) e_2 &= U_2 \bar R_{21} P_1^{-1} \ell_2^*(z_2) e_2 = U_2 \ell_2^*(z_2) e_2 P_1^{-1}.
\end{align*}
For the fourth line, we use the identity
\begin{equation*}
P_1^{-1} e_2^t \ell_2(z_2) = e_2^t \left( \ell_2(z_2) - \frac{\hbar}{z_2-Q_1} (\ell_2(z_2) - \ell_2(Q_1)) \right) \bar R_{21} P_1^{-1}. \qedhere
\end{equation*}
\end{proof}

\section{Quantum Hamiltonian reduction} \label{section:quantumHamiltonianReduction}

Quantum Hamiltonian reduction systematically quantizes the classical Hamiltonian reduction procedure and can be viewed as a special case of BRST quantization. We give an overview over the method of quantum Hamiltonian reduction in appendix \ref{appendix:reduction}. To perform the quantum Hamiltonian reduction determining the algebra $\mathfrak{A}_{N,\ell}$ which describes the observables of the rational spin Ruijsenaars--Schneider model, we extend the quantum cotangent bundle $\mathcal{O}_\hbar(T^* \mathrm{GL}_N)$ by tensoring with an algebra of canonical oscillators, in analogy with the classical case:

\begin{definition}
We let $\mathcal{O}_\hbar(T^* \C^{N \times \ell})$ be the unital associative algebra generated by the respective $N$ coefficients of $\ell$ column vectors $a^\alpha$ and of $\ell$ row vectors $b^\alpha$ ($\alpha=1,\dots,\ell$), subject to the canonical commutation relations $[a_1^\alpha,b_2^\beta] = -\hbar\delta^{\alpha\beta} \Delta_{12}$ and $[a_1^\alpha,a_2^\beta] = [b_1^\alpha,b_2^\beta] = 0$, where $\Delta = \sum_{i=1}^N e_i \otimes e_i^t$ and the subscripts label auxiliary spaces.
\end{definition}

We thus take our algebra of observables before reduction to be the tensor product of algebras $\mathcal{O}_\hbar(T^* \mathrm{GL}_N) \otimes \mathcal{O}_\hbar(T^* \C^{N \times \ell})$, which is generated by the coefficients of $T^{\pm 1},U^{\pm 1},Q,P^{\pm 1},a^\alpha,$ and $b^\alpha$. In particular, taking the tensor product implies that the generators $T^{\pm 1},U^{\pm 1},Q,P^{\pm 1}$ from the first factor commute with the generators $a^\alpha,b^\alpha$ from the second factor.  The next step in the reduction procedure is to identify the quantum moment map.

\subsection{Quantum moment map and induced action}

\begin{prop}
Let $E = (e_{ij})$ denote the matrix of matrix unit generators of $\mathfrak{gl}_N$. Then
\begin{align}
\mu \colon U(\mathfrak{gl}_N) \to \mathcal{O}_\hbar(T^* \mathrm{GL}_N) \otimes \mathcal{O}_\hbar(T^* \C^{N \times \ell}), \quad E \mapsto TQT^{-1} - UQU^{-1} - ab
\end{align}
is a homomorphism.
\end{prop}

\begin{proof}
The variables $x \coloneq TQT^{-1}$ and $x' \coloneq -UQU^{-1}$ satisfy the commutation relations of $\mathfrak{gl}_N$:
\begin{align*}
[x_1,x_2] = \hbar(C_{12} x_1 - x_1 C_{12}), \quad [x_1',x_2'] = \hbar(C_{12} x_1' - x_1' C_{12}).
\end{align*}
This is due to
\begin{align*}
(T_1 Q_1 T_1^{-1}) (T_2 Q_2 T_2^{-1})
&= T_2 T_1 R_{12} Q_1 R_{21} Q_2 R_{12} T_2^{-1} T_1^{-1} \\
&= (T_2 Q_2 T_2^{-1}) (T_1 Q_1 T_1^{-1}) + \hbar C_{12} ((T_1 Q_1 T_1^{-1})-(T_2 Q_2 T_2^{-1})), \\
(U_1 Q_1 U_1^{-1}) (U_2 Q_2 U_2^{-1})
&= U_2 U_1 R_{21} Q_1 R_{12} Q_2 R_{21} U_2^{-1} U_1^{-1} \\
&= (U_2 Q_2 U_2^{-1}) (U_1 Q_1 U_1^{-1}) - \hbar C_{12} ((U_1 Q_1 U_1^{-1})-(U_2 Q_2 U_2^{-1})),
\end{align*}
where we have used the identities
\begin{align*}
R_{12} Q_1 R_{21} Q_2 - Q_2 R_{12} Q_1 R_{21} &= \hbar C_{12} Q_1 R_{21} - \hbar R_{12} Q_1 C_{12}, \\
R_{21} Q_1 R_{12} Q_2 - Q_2 R_{21} Q_1 R_{12} &= -\hbar C_{12} Q_1 R_{12} + \hbar R_{21} Q_1 C_{12}.
\end{align*}
Furthermore, it is easily checked that $-ab$ also satisfies the commutation relations of $\mathfrak{gl}_N$ due to the Jordan-Schwinger representation:
\begin{align*}
[(ab)_1,(ab)_2] = -\hbar (C_{12} (ab)_1 - (ab)_1 C_{12}).
\end{align*}
Finally, we note that $x = TQT^{-1}$ and $x' = UQU^{-1}$ as well as $ab$ all commute among each other, making $\mu$ into a homomorphism.
\end{proof}

We now take $\mu$ to be our quantum moment map and perform quantum Hamiltonian reduction at the ideal generated by the coefficients of $E-\gamma \hbar I$. The reduced algebra $\mathfrak{A}_{N,\ell}$ resulting from quantum Hamiltonian reduction is then given by taking invariants and modding out the moment map equation:
\begin{equation}
\mathfrak{A}_{N,\ell} \coloneq \Big( (\mathcal{O}_\hbar(T^* \mathrm{GL}_N) \otimes \mathcal{O}_\hbar(T^* \C^{N \times \ell})) /(\mu(E)-\gamma\hbar) \Big)^{\mathfrak{gl}_N}.
\end{equation}
This means that that next step is to determine the algebra of invariants. We begin by determining the induced action:

\begin{lemma}\label{lemma:action}
The induced action of $U(\mathfrak{gl}_N)$ on $\mathcal{O}_\hbar(T^* \mathrm{GL}_N) \otimes \mathcal{O}_\hbar(T^* \C^{N \times \ell})$ is given by
\begin{equation}
\begin{aligned}
\operatorname{ad}_{E_1} T_2
&= -\hbar C_{12} T_2 + \hbar T_2 (T_1 X_{12} T_1^{-1}), \\
\operatorname{ad}_{E_1} U_2
&= -\hbar C_{12} U_2 + \hbar U_2 (U_1 X_{12} U_1^{-1}), \\
\operatorname{ad}_{E_1} P_2
&= \hbar P_2 (T_1 X_{12} T_1^{-1}) - \hbar (U_1 X_{12} U_1^{-1}) P_2,
\end{aligned}
\qquad
\begin{aligned}
\operatorname{ad}_{E_1} a_2 &= -\hbar C_{12} a_2, \\
\operatorname{ad}_{E_1} b_2 &= \hbar b_2 C_{12}, \\
\operatorname{ad}_{E_1} Q_2 &= 0,
\end{aligned}
\end{equation}
where $X_{12} \coloneq \sum_{i,j=1}^N e_{ij} \otimes e_{jj}$.
\end{lemma}

\begin{proof}
Any homomorphism from a Hopf algebra into another algebra (`quantum moment map') gives rise to an action on the codomain via the adjoint action, see appendix \ref{appendix:reduction}. In our case, this means that the action of the generating matrix $E$ of $U(\mathfrak{gl}_N)$ acts on $O \in \mathcal{O}_\hbar(T^* \mathrm{GL}_N) \otimes \mathcal{O}_\hbar(T^* \C^{N \times \ell})$ via
\begin{equation*}
\operatorname{ad}_E O \coloneq [\mu(E),O] = [TQT^{-1},O]-[UQU^{-1},O]-[ab,O].
\end{equation*}
The rest is a straightforward calculation using the identities
\begin{align}\label{eq:actionIdentity}
Q_1 R_{12}-R_{12} Q_1 &= \hbar (C_{12}-X_{12}), \\
Q_1 \bar R_{12}^{-1} - \bar R_{12}^{-1} Q_1 &= \hbar (X_{12}-Y_{12}),
\end{align}
as well as $Y_{12} \bar R_{12} = Y_{12}$, and
\begin{align*}
R_{12} X_{12} = X_{12}, \quad R_{21} X_{12} &= X_{12}, \quad X_{12} R_{12} = X_{12}, \quad X_{12} R_{21} = X_{12}, \\
\bar R_{12} X_{12} = X_{12}, \quad X_{12} \bar R_{12} &= X_{12}, \quad X_{12} \bar R_{21} = X_{12}, \quad X_{12} C_{12} = X_{12}.
\end{align*}
For example, for $T$, we can use equation \eqref{eq:actionIdentity} to derive
\begin{align*}
\operatorname{ad}_{E_1} T_2
&= [T_1 Q_1 T_1^{-1},T_2] \\
&= -\hbar T_2 T_1 (C_{12}-X_{12}) R_{21} T_1^{-1} \\
&= -\hbar C_{12} T_2 + \hbar T_2 (T_1 X_{12} T_1^{-1}).
\qedhere
\end{align*}
\end{proof}

\begin{remark}
The action of a Lie algebra element $M \in \mathfrak{gl}_N$ can be computed as follows
\begin{equation}
\operatorname{ad}_M T = \operatorname{Tr}_1(M_1 \operatorname{ad}_{E_1} T_2)
= - \hbar M T + \hbar T \sum_j e_{jj} (e_j^t T^{-1} M e)
\end{equation}
We remark that for torus elements of the form $M = TDT^{-1}$ with $D$ diagonal, the action simplifies drastically due to
\begin{equation}
\hbar T \sum_j e_{jj} (e_j^t T^{-1} M e) = \hbar TD.
\end{equation}
Furthermore, the action of any $M \in \mathfrak{gl}_N$ indeed preserves the mirabolic condition for $T$ as
\begin{equation}
\begin{aligned}
\operatorname{ad}_M Te
&= - \hbar M e + \hbar T \sum_j e_j e_j^t T^{-1} M e = 0.
\end{aligned}
\end{equation}
In this sense, we may say that the $\mathfrak{m}_N$-gauge symmetry is preserved at the quantum level. Furthermore, when $M$ lies in the mirabolic subalgebra, i.e. $Me = 0$, then $\operatorname{ad}_M T$ becomes equal to $\hbar MT$ as expected. More generally, the terms in lemma \ref{lemma:action} that include $T_1 X_{12} T_1^{-1}$ and $U_1 X_{12} U_1^{-1}$ drop out when we restrict to the action of $\mathfrak{m}_N$ alone. In particular, $P$ is invariant under the action of the mirabolic subalgebra and $T,U,a$ transform in the vector representation of $\mathfrak{m}_N$, while $T^{-1},U^{-1},b$ transform in the covector representation.
\end{remark}

\subsection{Algebra of mirabolic invariants}

The next step is to determine the algebra of $\mathfrak{gl}_N$-invariants, i.e. the commutant of the moment map. To this end, we factorize the group $\mathrm{GL}_N$ as $\mathrm{GL}_N = \mathrm{T}_N \mathrm{M}_N$, where $\mathrm{T}_N$ is the maximal torus of diagonal invertible matrices and $\mathrm{M}_N$ is the mirabolic subgroup. This allows us to perform the reduction in two steps, first with respect to the mirabolic subgroup $\mathrm{M}_N$, and then with respect to the torus $\mathrm{T}_N$, and it will turn out to be interesting to consider mirabolic invariants in their own right. To this end, let us now describe the algebra of observables that is invariant under the mirabolic subgroup $\mathrm{M}_N$. Such invariants can be expressed via $Q$ and $P$ as well as the following variables:

\begin{definition}
The \emph{dressed oscillators} are defined to be $A \coloneq T^{-1} a, B \coloneq bU$ and $C \coloneq bUP$ while the \emph{Frobenius matrix} and the \emph{Lax matrix} are defined to be $W \coloneq T^{-1} U$ and $L \coloneq T^{-1} UP$, respectively. We also define $\bar A \coloneq C L^{-1}$ and $\bar C \coloneq L^{-1} A$ for later purposes. It will be convenient for us to express these generators as polynomials as follows:
\begin{align}
A^\alpha(z) \coloneq e^t \ell(z) A^\alpha, \quad C^\alpha(z) \coloneq C^\alpha \ell^*(z) e, \quad L(z,z') \coloneq e^t \ell(z) L \ell^*(z') e.
\end{align}
\end{definition}

\begin{prop}
The algebra $\left( \mathcal{O}_\hbar(T^* \mathrm{GL}_N) \otimes \mathcal{O}_\hbar(T^* \C^{N \times \ell}) \right)^{\mathfrak{m}_N}$ of operators invariant under $\mathfrak{m}_N$ is generated by the coefficients of $Q,A,C,P^{\pm 1}$, and $L^{\pm 1}$, subject to the relations $C_1^\alpha Q_2 = (Q_2-\hbar Y_{12}) C_1^\alpha$ and $L_1 Q_2 = (Q_2-\hbar Y_{12}) L_1$ as well as
\begin{equation}
\begin{aligned}
A_1^\alpha A_2^\beta &= R_{12} A_2^\beta A_1^\alpha, \\
C_1^\alpha C_2^\beta \bar R_{12}^{-1} &= C_2^\beta C_1^\alpha \bar R_{21}^{-1} R_{21}, \\
L_1 A_2^\alpha &= R_{12} A_2^\alpha L_1,
\end{aligned}
\qquad
\begin{aligned}
A_1^\alpha C_2^\beta + \hbar\delta^{\alpha\beta} \Delta_{12} L_2 &= C_2^\beta \bar R_{12}^{-1} A_1^\alpha, \\
L_1 \bar R_{21}^{-1} L_2 \bar R_{12}^{-1} R_{12} &= R_{12} L_2 \bar R_{12}^{-1} L_1 \bar R_{21}^{-1}, \\
C_1^\alpha L_2 \bar R_{12}^{-1} &= L_2 C_1^\alpha \bar R_{21}^{-1} R_{21},
\end{aligned}
\end{equation}
with all other pairs commuting. In terms of polynomials, we can write
\begin{align}
A^\alpha(z_1) A^\beta(z_2) ={}& \mathbf{R}_{12} A^\beta(z_2) A^\alpha(z_1), \\
A^\alpha(z_1) C^\beta(z_2) + \hbar\delta^{\alpha\beta} L(z_1,z_2) ={}& \chi(z_2) \mathbf{R}_{12}^*[\chi(z_2)^{-1} C^\beta(z_2)] A^\alpha(z_1) + \hbar \chi(z_2) \partial_{12} \mathbf{S}^{\alpha\beta}(z_1), \\
C^\alpha(z_1) C^\beta(z_2) ={}& \mathbf{R}_{12}^* C^\beta(z_2) C^\alpha(z_1),
\end{align}
where $\mathbf{S}^{\alpha\beta}(z) \coloneq C^\alpha (z-Q)^{-1} A^\beta$.
\end{prop}

\begin{proof}
The first six relations then follow from a straightforward application of the commutation relations of $\mathcal{O}_\hbar(T^* \mathrm{GL}_N)$, and we can rewrite them in terms of polynomials using the identities
\begin{align*}
R_{21} \ell_1^*(z_1) e_1 \ell_2^*(z_2) e_2 = \mathbf{R}_{21}^* \ell_1^*(z_1) e_1 \ell_2^*(z_2) e_2, \quad \ell_2^*(z) e_2 C_1^\alpha = C_1^\alpha \bar R_{21}^{-1} \ell_2^*(z) e_2,
\end{align*}
and
\begin{align*}
e_2^t \ell_2(z_2) C_1^\alpha \bar R_{21}^{-1} = C_1^\alpha e_2^t \ell_2(z_2) - \hbar C_1^\alpha (z_2-Q_1)^{-1} e_2^t (\ell_2(z_2) - Y_{12}).
\end{align*}
With this, we can derive the cross relation in polynomial form using:
\begin{align*}
e_2^t \ell_2(z_2) C_1^\alpha \bar R_{21}^{-1} A_2^\beta \ell_1^*(z_1) e_1
={}& \mathbf{R}_{21}^*[C_1^\alpha (z_1-Q_1)^{-1} e_1] \chi(z_1) e_2^t \ell_2(z_2) A_2^\beta + \hbar \partial_{12}[\mathbf{S}^{\alpha\beta}(z_2)] \chi(z_1),
\end{align*}
while the relations for $C^\alpha$ in the dual Lagrange basis become
\begin{align*}
C_1^\alpha C_2^\beta \bar R_{12}^{-1} \ell_1^*(z_1) e_1 \ell_2^*(z_2) e_2 = C_2^\beta C_1^\alpha \bar R_{21}^{-1} R_{21} \ell_1^*(z_1) e_1 \ell_2^*(z_2) e_2,
\end{align*}
which using the relations above simplifies to
\begin{equation*}
C^\alpha(z_1) C^\beta(z_2) = \mathbf{R}_{12}^* C^\beta(z_2) C^\alpha(z_1). \qedhere
\end{equation*}
\end{proof}

\begin{remark}
We briefly note that the commutation relations for $L^{-1}$ follow from the commutation relations of $L$. In addition, we have $L^{-1}e = P^{-1} U^{-1} Te = P^{-1}e$ and $P^{-1}$ is diagonal, the coefficients of $P^{-1} e$ are the diagonal entries of $P^{-1}$, which means that the commutation relations for $P^{\pm 1}$ also follow from the commutation relations of $L$.
\end{remark}

An important property of the dressed oscillators $C^\alpha(z)$ is that they satisfy the exchange relation in spin space:

\begin{prop}
The commutation relations
\begin{equation}
C^\alpha(z_1) C^\beta(z_2) = \mathbf{R}_{12}^* C^\beta(z_2) C^\alpha(z_1)    
\end{equation}
may equivalently be written as
\begin{equation}\label{eq:exchangeRel}
C^1(z_1) C^2(z_2) = \hat R^{12}(z_1-z_2) C^2(z_2) C^1(z_1),
\end{equation}
where $\hat R^{12}(z) = \frac{z}{z+\hbar} + \frac{\hbar}{z+\hbar} P^{12}$ is Yang's unitary $R$-matrix acting in spin space $\C^\ell \otimes \C^\ell$. The upper indices in equation \eqref{eq:exchangeRel} denote auxiliary spin spaces.
\end{prop}

\begin{remark}
Equation \eqref{eq:exchangeRel} forms one half of a Faddeev--Zamolodchikov algebra \cite{zamolodchikov:1978,faddeev:1980}. By Faddeev--Zamolod\-chikov algebra, we mean the algebra generated by a vector $C$ and covector $C^\dagger$ satisfying a generalization of canonical commutation relations. These have the form \eqref{eq:exchangeRel} together with its adjoint for $C^\dagger$ and the cross relation
\begin{equation}
C^1(z_1) C^{2\dagger}(z_2) = C^{2\dagger}(z_2) \hat R^{12}(z_1-z_2) C^1(z_1) + \Delta^{12} \delta(z_1-z_2)    
\end{equation}
where $\Delta^{12} = \sum_\alpha e_\alpha^t \otimes e_\alpha$ and $\delta(z_1-z_2)$ is the Dirac delta function. A word of caution: While the dressed oscillators $A^\alpha(z)$ are natural candidates to play the role of $C^\dagger$, they do not satisfy the exchange relation for any $R$-matrix acting in spin space, as the operator $\mathbf{R}_{12} A^\beta(z_2) A^\alpha(z_1)$ includes terms with coincident arguments such as $A^\beta(z_1) A^\alpha(z_1)$ and $A^\beta(z_2) A^\alpha(z_2)$. Thus, unfortunately $A^\alpha(z)$ and $C^\alpha(z)$ do not form a Faddeev--Zamolodchikov algebra together.
\end{remark}

\subsection{Fixing the residual torus gauge}

In the preceding section, we have found the subalgebra of invariants under the action of the mirabolic subgroup. It remains to consider the residual gauge symmetry of the torus. There are several ways to go about finding invariant variables under the torus action. Let us try to proceed in analogy with the classical case and adjoin the inverses $t^{-1}$ to the algebra, with $t$ the unique diagonal matrices satisfying $te = \sum_{\rho=1}^\ell A^\rho$. In this case, we would want to build the fully gauge invariant combinations
\begin{align}
\mathbf{a}^\alpha \coloneq t^{-1} A^\alpha, \quad \mathbf{c}^\alpha \coloneq C^\alpha t, \quad \mathbf{L} \coloneq t^{-1} L t.
\end{align}
However, when we try to compute the commutation relations for these variables, we see that the components $t_i$ of $t$ commute with their inverses according to
\begin{equation}
[t_i^{-1},t_j] = \frac{\hbar}{q_{ij}} (1-t_j t_i^{-1}-t_i^{-1} t_j + t_i^{-1} t_j^2 t_i^{-1})
\end{equation}
and notice that this cannot be rearranged in a way that allows a reordering making $t_i$ and $t_i^{-1}$ into generators of a Poincaré--Birkhoff--Witt basis. The failure of the Poincaré--Birkhoff--Witt property is a strong indication that the variables $\mathbf{a}$ that contain $t^{-1}$ are problematic for quantization. In particular, there is no canonical ordering that makes cancelations manifest when computing commutators containing $\mathbf{a}$.

We can alternatively go the route of imposing a gauge fixing constraint via the Dirac method, and there are two interesting constraints to consider. Firstly, we note that the commutation relations for $A$ may be written in component form as
\begin{equation}
[A_i^\alpha,A_j^\beta] = -\frac{\hbar}{q_{ij}} (A_i^\alpha-A_j^\alpha) (A_i^\beta-A_j^\beta).
\end{equation}
From this it is clear that $\sum_{\rho=1}^\ell A_i^\rho = 1$ for $i=1,\dots,N$ is in fact a system of first-class constraints and that they give a complete gauge fixing of the torus action on $A$. However, the computations become slightly more transparent if we instead impose the gauge fixing $A_i^\ell = 1$, i.e. we impose that the $\ell$th spin components are fixed to one. Let $J$ be the right ideal generated by $A_i^\ell - 1$ for $i=1,\dots,N$. This is also a 
system of first-class constraints, i.e. $[J,J] \subseteq J$, and importantly satisfies
\begin{equation}
[J,A_i^\alpha] \subseteq J, \quad [J,C_i^\alpha] \subseteq J - \hbar\delta^{\ell\alpha} ({\cdots}),
\end{equation}
so as long as $\alpha < \ell$ we also have $[J,C_i^\alpha] \subseteq J$. According to the Dirac method, physical observables $O$ are characterized by the condition $[J,O] \subseteq J$. It follows that $q_i,A_i^\alpha,C_i^\alpha$ for $i=1,\dots,N$ and $\alpha = 1,\dots,\ell-1$ are all physical observables under this gauge fixing, i.e. we have essentially eliminated one unphysical polarization from the spins. We can solve for the unphysical $B_i^\ell$ via the moment map equation and the mirabolic condition for $W$, which gives
\begin{equation}
1 = \sum_j W_{ij} = A_i \sum_j \frac{1}{q_{ij}+\gamma\hbar} B_j = \sum_j \frac{1}{q_{ij}+\gamma\hbar} B_j^\ell + \sum_{\rho=1}^{\ell-1} A_i^\rho \sum_j \frac{1}{q_{ij}+\gamma\hbar} B_j^\rho.
\end{equation}
Letting $\Gamma_{ij} = \frac{1}{q_{ij}+\gamma\hbar}$ be the Cauchy matrix, we can write
\begin{equation}
B_i^\ell = \sum_j \Gamma_{ij}^{-1} - \sum_{\rho=1}^{\ell-1} \sum_{jk} \Gamma_{ij}^{-1} A_j^\rho \Gamma_{jk} B_k^\rho,
\end{equation}
which implies
\begin{equation}
W_{ij} = \bar W_{ij} + W_{ij}^\text{red} - \bar W_{ij} (\bar W^{-1} W^\text{red} e)_j,
\end{equation}
where
\begin{equation}
\bar W_{ij} \coloneq \frac{\prod_{a(\neq i)} (q_{aj}+\gamma\hbar)}{\prod_{a(\neq j)} q_{aj}}, \quad W_{ij}^\text{red} \coloneq \sum_{\rho=1}^{\ell-1} A_i^\rho \Gamma_{ij} B_j^\rho.
\end{equation}
Note that $\bar W$ coincides with the spinless case \cite{arutyunov:1996}. The commutation relations between $A$ and $B$ can then be written as
\begin{equation}\label{eq:commutatorAB}
\begin{aligned}
[A_i^\alpha,B_j^\beta]
={}& -\hbar\delta^{\alpha\beta} W_{ij} \\
= &-\hbar\delta^{\alpha\beta} \Gamma_{ij} \sum_k \Gamma_{jk}^{-1} - \hbar\delta^{\alpha\beta} \sum_{\rho=1}^{\ell-1} \sum_{kl} (\delta_{ik} \delta_{jl} - \Gamma_{ij} \Gamma_{jk}^{-1}) A_k^\rho \Gamma_{kl} B_l^\rho \\
= &-\hbar\delta^{\alpha\beta} \bar W_{ij} - \hbar\delta^{\alpha\beta} W_{ij}^\text{red} + \hbar\delta^{\alpha\beta} \bar W_{ij} (\bar W^{-1} W^\text{red} e)_j
\end{aligned}
\end{equation}
and a sanity check reveals that the Poisson bracket that is the classical limit of these commutators does indeed satisfy the Jacobi identity.

We would like to rewrite the commutation relations \eqref{eq:commutatorAB} in polynomial form, and we can see that
\begin{equation}
\sum_{i,j=1}^N \ell_i(z_1) \bar W_{ij} \ell_j^*(z_2+\gamma\hbar) = \delta(z_1,z_2).
\end{equation}
However, it is unclear to us how the other terms in \eqref{eq:commutatorAB} can be simplified in polynomial form. Instead, we can consider the alternative gauge fixing that imposes $e^t \chi'(Q)^{-1} W = e^t \chi'(Q)^{-1}$. This is indeed a first-class constraint. To see this, note that $W$ satisfies the commutation relation
\begin{equation}
R_{12} W_1 W_2 = W_2 W_1 R_{12}
\end{equation}
and that $(e^t \chi'(Q)^{-1} \otimes e^t \chi'(Q)^{-1}) R_{12} = (e^t \chi'(Q)^{-1} \otimes e^t \chi'(Q)^{-1})$. The advantage of this constraint is that the moment map equation, which in terms of $W$ reads $QW-WQ+\gamma\hbar W = AB$, can be cast into a particularly nice polynomial form. Specifically, we use the identities
\begin{align}
\ell(z) Q = z \ell(z) - \chi(z) \chi'(Q)^{-1}, \quad Q \ell^*(z) = z \ell^*(z) - \chi(z)
\end{align}
to obtain
\begin{equation}
\begin{aligned}
e^t \ell(z) AB \ell^*(z') e
={}& e^t \ell(z) (QW-WQ+\gamma\hbar) \ell^*(z') e \\
={}& (z-z'+\gamma\hbar) e^t \ell(z) W \ell^*(z') e - \chi(z) e^t \chi'(Q)^{-1} W \ell^*(z') e + \chi(z') e^t \ell(z) W e \\
={}& (z-z'+\gamma\hbar) e^t \ell(z) W \ell^*(z') e - \chi(z) + \chi(z')
\end{aligned}
\end{equation}
Introducing $B(z) \coloneq B \ell^*(z) e$ and $W(z,z') \coloneq e^t \ell(z) W \ell^*(z') e$ and noticing that $\chi(z) - \chi(z') = (z-z') \delta(z,z')$ with $\delta(z,z') \coloneq \sum_{i=1}^N \ell_i(z) \ell_i^*(z')$, we may thus write
\begin{equation}
A(z) B(z') = (z-z'+\gamma\hbar) W(z,z') - (z-z') \delta(z,z'),
\end{equation}
which is a particularly elegant form of the moment map equation. In particular, it implies
\begin{equation}
A(z) B(z+\gamma\hbar) = \gamma\hbar \delta(z,z+\gamma\hbar).
\end{equation}
Further, noting that $A$ and $B$ commute according to $A_1^\alpha B_2^\beta + \hbar \delta^{\alpha\beta} \Delta_{12} W_2 = B_2^\beta A_1^\alpha$, we can write down their commutation relation in polynomial form:
\begin{equation}
A^\alpha(z_1) B^\beta(z_2) + \hbar \delta^{\alpha\beta} \sum_{\rho=1}^\ell \frac{A^\rho(z_1) B^\rho(z_2)}{z_1-z_2+\gamma\hbar} = B^\beta(z_2) A^\alpha(z_1) + \frac{z_1-z_2}{z_1-z_2+\gamma\hbar} \delta(z_1,z_2),
\end{equation}
which may be written in auxiliary space notation as a Faddeev--Zamolodchikov-type relation:
\begin{equation}
A^1(z_1) \hat R^{12}(z_1-z_2+\gamma\hbar) B^2(z_2) = B^2(z_2) A^1(z_1) + \frac{z_1-z_2}{z_1-z_2+\gamma\hbar} \delta(z_1,z_2),
\end{equation}
where we used $\hat R^{12}(z) = 1 + \frac{\hbar}{z} P^{12}$. However, despite these suggestive formulas, it is at this moment unclear to us how to find a good set of generators of the algebra of physical operators corresponding to this gauge fixing.

\section{The reduced algebra} \label{section:reducedAlgebra}

Besides the gauge fixing procedures above, we can also perform the reduction with respect to the residual torus in a fully gauge covariant way by simply building combinations that are products of strings of $L^{\pm 1}$ and $Q$ sandwiched between $C$ and $A$, e.g. $C L^{-1} Q L L A C Q A$. Such combinations are manifestly invariant under the residual torus symmetry (and by construction also under the mirabolic subgroup). Importantly, we can use the moment map equation, which may be cast as $[Q,L] = AC-\gamma\hbar L$, to bring these strings into a canonical form where $L^{\pm 1}$ always stands to the left of $Q$. A generating function for these canonical forms may be given by
\begin{align}
\mathbf{S}[n](z) \coloneq C L^{n-1} (z-Q)^{-1} A = \sum_{r=0}^\infty z^{-r-1} C L^{n-1} Q^r A.
\end{align}
We can further use the moment map equation to obtain a recurrence relation expressing $\mathbf{S}[n](z)$ in terms of $\mathbf{S}[n-1](z)$ and $\mathbf{S}[1](z)$ for any positive $n$:

\begin{prop}
The generating functions $\mathbf{S}[n](z)$ with $n > 1$ satisfy the recurrence relation
\begin{align}
\mathbf{S}[n](z) = \oint_\mathcal{C} d\tau \mathbf{S}[n-1](\tau) \partial_{\tau+\gamma\hbar,z} \mathbf{S}[1](z),
\end{align}
where $\mathcal{C}$ is a contour encircling $q_1,\dots,q_N$ counterclockwise, but no other poles. Here we have used the notation $\partial_{w,z} f(z) \coloneq \frac{f(w)-f(z)}{w-z}$.
\end{prop}

\begin{proof}
We first note that the moment map equation in components reads
\begin{equation*}
q_i L_{ij} - L_{ij} q_j + \gamma\hbar L_{ij} = A_i C_j,
\end{equation*}
which implies $(q_{ij}+(\gamma+1)\hbar) L_{ij} = A_i C_j$, so
\begin{equation*}
L_{ij} = A_i \frac{1}{q_{ij}+(\gamma+1)\hbar} C_j.
\end{equation*}
Due to $\frac{1}{q_{ij}+(\gamma+1)\hbar} = \oint_\mathcal{C} d\tau (\tau-q_i)^{-1} (\tau-q_j+(\gamma+1)\hbar)^{-1}$, we find
\begin{equation*}
L = \oint_\mathcal{C} d\tau (\tau-Q)^{-1} AC (\tau+\gamma\hbar-Q)^{-1}.
\end{equation*}
Using $(\tau+\gamma\hbar-Q)^{-1} (z-Q)^{-1} = \partial_{\tau+\gamma\hbar,z} (z-Q)^{-1}$, it follows that
\begin{align*}
\mathbf{S}[n](z)
&= CL^{n-1} (z-Q)^{-1} A \\
&= \oint_\mathcal{C} d\tau CL^{n-2} (\tau-Q)^{-1} AC (\tau+\gamma\hbar-Q)^{-1} (z-Q)^{-1} A \\
&= \oint_\mathcal{C} d\tau \mathbf{S}[n-1](\tau) \partial_{\tau+\gamma\hbar,z} \mathbf{S}[1](z).\qedhere
\end{align*}
\end{proof}

\subsection{Quantum current}

The recurrence relation above implies that the positive modes $\mathbf{S}[n](z)$ can be expressed in terms of $\mathbf{S}(z) \coloneq \mathbf{S}[1](z)$, which in turn can be expressed in terms of its residues
\begin{equation}
\mathbf{S}_i \coloneq C_i A_i.
\end{equation}
We call $\mathbf{S}(z)$ the \emph{quantum current}, stemming from the fact that the $\mathbf{S}_i$ satisfy the following commutation relations that first appeared in \cite{reshetikhin:1990} under that name. It can be checked that the classical limit of these commutation relations coincides with the Poisson bracket determined in \cite{arutyunov:1998} as expected.

\begin{prop}
For any $\epsilon \in \C^\times$, we have
\begin{align}\label{eq:SSrelation}
\hat R^{21}(q_{ji}) \mathbf{S}_i^1 \hat R^{12}(q_{ij}+\gamma\hbar) \mathbf{S}_j^2 &= \mathbf{S}_j^2 \hat R^{21}(q_{ji}+\gamma\hbar) \mathbf{S}_i^1 \hat R^{12}(q_{ij}), \quad (i \neq j) \\ \label{eq:SSrelation2}
\hat R^{21}(\epsilon) \mathbf{S}_i^1 \hat R^{12}(\gamma\hbar) \mathbf{S}_i^2 &= \mathbf{S}_i^2 \hat R^{21}(\gamma\hbar) \mathbf{S}_i^1 \hat R^{12}(\epsilon),
\end{align}
where $\hat R^{12}(z) = 1+\frac{\hbar}{z} P^{12}$ is Yang's rational $R$-matrix acting in spin space $\C^\ell$ and we have used auxiliary space notation.
\end{prop}

\begin{proof}
Looking at the coefficients individually, we derive
\begin{align}
[\mathbf{S}_i^{\alpha\beta},\mathbf{S}_j^{\mu\nu}]
={}& \hbar \sum_{\rho=1}^\ell \mathbf{S}_j^{\mu\rho} \frac{\delta^{\alpha\nu}}{q_{ji}+\gamma\hbar} \mathbf{S}_i^{\rho\beta} - \hbar \sum_{\rho=1}^\ell \mathbf{S}_i^{\alpha\rho} \frac{\delta^{\mu\beta}}{q_{ij}+\gamma\hbar} \mathbf{S}_j^{\rho\nu} \label{eq:SS1} \\
&+ \frac{\hbar}{q_{ij}+\epsilon\delta_{ij}} \bigg( \mathbf{S}_i^{\mu\beta} \mathbf{S}_j^{\alpha\nu} + \hbar \sum_{\rho=1}^\ell \mathbf{S}_i^{\mu\rho} \frac{\delta^{\alpha\beta}}{q_{ij}+\gamma\hbar} \mathbf{S}_j^{\rho\nu} \bigg) \label{eq:SS2} \\
&- \frac{\hbar}{q_{ji}+\epsilon\delta_{ij}} \bigg( \mathbf{S}_j^{\mu\beta} \mathbf{S}_i^{\alpha\nu} + \hbar \sum_{\rho=1}^\ell \mathbf{S}_j^{\mu\rho} \frac{\delta^{\alpha\beta}}{q_{ji}+\gamma\hbar} \mathbf{S}_i^{\rho\nu} \bigg) \label{eq:SS3}
\end{align}
making use of the identity
\begin{align*}
\bar R_{21}^{-1} R_{21} \bar R_{12} (e_{ii})_1 \bar R_{12}^{-1} (e_{jj})_2 R_{12}
={} (e_{jj})_2 \bar R_{21}^{-1} (e_{ii})_1 &+ \frac{\hbar}{q_{ij}+\epsilon\delta_{ij}} (e_{ii})_2 \bar R_{21}^{-1} (e_{jj})_1 C_{12} R_{12} \\
&- \frac{\hbar}{q_{ji}+\epsilon\delta_{ij}} (e_{jj})_2 \bar R_{21}^{-1} (e_{ii})_1 C_{12} R_{12},
\end{align*}
which holds for any $\epsilon \in \C^\times$. The relation can then be rewritten in $R$-matrix form as claimed.
\end{proof}

\begin{remark}
Let us examine skew-symmetry in the formula for the commutator $[\mathbf{S}_i^{\alpha\beta},\mathbf{S}_j^{\mu\nu}]$. The line \eqref{eq:SS1} gives something manifestly skew-symmetric, but the lines \eqref{eq:SS2} and \eqref{eq:SS3} do not. However, we have the identity
\begin{equation}
\begin{aligned}
&\left( (e_{jj})_2 \bar R_{21}^{-1} (e_{ii})_1 + (e_{ii})_2 \bar R_{21}^{-1} (e_{jj})_1 \right) C_{12} R_{12} \\
&= \bar R_{21}^{-1} R_{21} \bar R_{12} C_{12} \left( (e_{ii})_2 \bar R_{21}^{-1} (e_{jj})_1 + (e_{jj})_2 \bar R_{21}^{-1} (e_{ii})_1 \right),
\end{aligned}
\end{equation}
giving the additional relation
\begin{equation}
\begin{aligned}
&\mathbf{S}_j^{\mu\beta} \mathbf{S}_i^{\alpha\nu} + \hbar \sum_{\rho=1}^\ell \mathbf{S}_j^{\mu\rho} \frac{\delta^{\alpha\beta}}{q_{ji}+\hbar\gamma} \mathbf{S}_i^{\rho\nu} + \mathbf{S}_i^{\mu\beta} \mathbf{S}_j^{\alpha\nu} + \hbar \sum_{\rho=1}^\ell \mathbf{S}_i^{\mu\rho} \frac{\delta^{\alpha\beta}}{q_{ij}+\hbar\gamma} \mathbf{S}_j^{\rho\nu} \\
&= \mathbf{S}_i^{\alpha\nu} \mathbf{S}_j^{\mu\beta} + \hbar \sum_{\rho=1}^\ell \mathbf{S}_i^{\alpha\rho} \frac{\delta^{\mu\nu}}{q_{ij}+\hbar\gamma} \mathbf{S}_j^{\rho\beta} + \mathbf{S}_j^{\alpha\nu} \mathbf{S}_i^{\mu\beta} + \hbar \sum_{\rho=1}^\ell \mathbf{S}_j^{\alpha\rho} \frac{\delta^{\mu\nu}}{q_{ji}+\hbar\gamma} \mathbf{S}_i^{\rho\beta},
\end{aligned}
\end{equation}
which is equivalent to skew-symmetry of the commutator.
\end{remark}

\begin{remark}
We can now present the algebra of positive modes via $q_i$ and $\mathbf{S}_i^{\alpha\beta}$ subject to the commutation relations $[q_i, \mathbf{S}_j^{\alpha\beta}] = \hbar\delta_{ij} \mathbf{S}_j^{\alpha\beta}$ as well as (\ref{eq:SSrelation}) and (\ref{eq:SSrelation2}). We remark that the moment map equation for the variables $\mathbf{a},\mathbf{c},\mathbf{L}$ can be written as $Q\mathbf{L}-\mathbf{L}Q+\gamma\hbar\mathbf{L} = \mathbf{a}\mathbf{c}$, which can be solved in the same way as before to express $\mathbf{L}$ in terms of $\mathbf{a}$ and $\mathbf{c}$:
\begin{align}
\mathbf{L} = \oint_\mathcal{C} d\tau (\tau-Q)^{-1} \mathbf{a} \mathbf{c} (\tau+\gamma\hbar-Q)^{-1},
\end{align}
where again $\mathcal{C}$ is a contour encircling $q_1,\dots,q_N$ counterclockwise. It follows that the variables $\mathbf{a},\mathbf{c}$, and $\mathbf{L}$ can all be expressed in terms of $\mathbf{S}_i^{\alpha\beta}$ after imposing the moment map equation and adjoining certain inverses to the algebra, since
\begin{align}
\mathbf{c}_i^\alpha = C_i^\alpha t_i = \sum_{\rho=1}^\ell \mathbf{S}_i^{\alpha\rho}, \quad \mathbf{a}_i^\alpha = t_i^{-1} A_i^\alpha = (\mathbf{c}_i^\alpha)^{-1} \mathbf{S}_i^{\alpha\alpha}.
\end{align}
Here, we notice again that the variables $\mathbf{a}_i^\alpha$ are problematic for quantization because they can only be expressed via inverses that have to be adjoined to the algebra, which generally destroys the Poincaré--Birkhoff--Witt property.
\end{remark}

\subsection{Loop algebra and Yangian}

Besides the quantum current, which may be seen as the basic building block from which we can build up the positive modes of the algebra $\mathfrak{A}_{N,\ell}$ of fully gauge-invariant observables, we can also find two important symmetry algebras within $\mathfrak{A}_{N,\ell}$: The loop algebra $L(\mathfrak{gl}_\ell)$ and the Yangian $Y(\mathfrak{gl}_\ell)$. The center of the loop algebra turns out to supply us with an infinite tower of commuting Hamiltonians, proving the integrability of the model, while the rest of the loop algebra equips the model with additional non-abelian integrals of motion. In contrast, the Yangian does not commute with the loop algebra and may be used to build up the spectrum of the commuting Hamiltonians via raising and lowering operators.

\begin{definition}
Introduce the full invariants
\begin{align}
\mathbf{J}[n] &\coloneq b g^n a = C L^{n-1} A = \oint dz \mathbf{S}[n](z), \\
\mathbf{T}(z) &\coloneq 1 + b (z-x)^{-1} a = 1 + \bar A (z-Q)^{-1} A, \\
\bar{\mathbf{T}}(z) &\coloneq 1 - b (z+\hbar-gxg^{-1})^{-1} a = 1 - C (z-Q)^{-1} \bar C.
\end{align}
\end{definition}

\begin{prop}
The coefficients $\mathbf{J}^{\alpha\beta}[n]$ satisfy the relations for the loop algebra $L(\mathfrak{gl}_\ell)$.
\end{prop}

\begin{proof}
It may be quickly checked that $\mathbf{J}^{\alpha\beta}[0] = b^\alpha a^\beta$ satisfy the commutation relations for the Lie algebra $\mathfrak{gl}_\ell$. Furthermore,
\begin{align*}
\mathbf{J}^{\alpha\beta}[n]& \mathbf{J}^{\mu\nu}[1] \\
={}& b_1^\alpha (UP)_1 (T^{-1}UP)_1^{n-1} T_1^{-1} a_1^\beta b_2^\mu (UP)_2 T_2^{-1} a_2^\nu \\
={}& b_2^\mu b_1^\alpha (UP)_2 (UP)_1 \bar R_{21}^{-1} R_{21} \bar R_{12} (\bar R_{12}^{-1} T_1^{-1} U_1 P_1 \bar R_{21}^{-1} R_{21} \bar R_{12})^{n-1} \bar R_{12}^{-1} R_{12} T_2^{-1} T_1^{-1} a_2^\nu a_1^\beta \\
&+ b_1^\alpha (UP)_1 (T_1^{-1} U_1 P_1)^{n-1} T_1^{-1} [a_1^\beta,b_2^\mu] (UP)_2 T_2^{-1} a_2^\nu \\
={}& \mathbf{J}^{\mu\nu}[1] \mathbf{J}^{\alpha\beta}[n] + \hbar \delta^{\alpha\nu} \mathbf{J}^{\mu\beta}[n+1] - \hbar \delta^{\mu\beta} \mathbf{J}^{\alpha\nu}[n+1]. \qedhere
\end{align*}
\end{proof}

\begin{prop}
The coefficients $\mathbf{T}^{\alpha\beta}(z)$
satisfy the relations for the Yangian $Y(\mathfrak{gl}_\ell)$:
\begin{align}
\mathbf{T}^{\alpha\beta}(z_1) \mathbf{T}^{\mu\nu}(z_2) &= \mathbf{T}^{\mu\nu}(z_2) \mathbf{T}^{\alpha\beta}(z_1) + \hbar\partial_{12}\mathbf{T}^{\mu\beta}(z_2) \mathbf{T}^{\alpha\nu}(z_1)
\end{align}
Furthermore, since $\prod_{i=1}^N (z-q_i) \mathbf{T}(z)$ is polynomial of degree $N$, we obtain an algebra homomorphism from the $N$-truncated Yangian $Y^N(\mathfrak{gl}_\ell)$ into $\mathfrak{A}_{N,\ell}$.
\end{prop}

\begin{proof}
For $\mathbf{T}^{\alpha\beta}(z)$, we derive
\begin{align*}
\mathbf{T}^{\alpha\beta}(z_1) \mathbf{T}^{\mu\nu}(z_2)
={}& \delta^{\alpha\beta} \delta^{\mu\nu} + \delta^{\alpha\beta} \bar A^\mu (z_2-Q)^{-1} A^\nu + \bar A^\alpha (z_1-Q)^{-1} A^\beta \\
&+ \bar A_2^\mu \bar A_1^\alpha R_{12} (z_1-Q_1)^{-1} R_{21} (z_2-Q_2)^{-1} R_{12} A_2^\nu A_1^\beta \\
&- \hbar \delta^{\mu\beta} \bar A^\alpha (z_1-Q)^{-1} (z_2-Q)^{-1} A^\nu \\
={}& \mathbf{T}^{\mu\nu}(z_2) \mathbf{T}^{\alpha\beta}(z_1) + \hbar \partial_{12} \mathbf{T}^{\mu\beta}(z_1) \mathbf{T}^{\alpha\nu}(z_2)
\end{align*}
due to $(z_1-Q)^{-1} (z_2-Q)^{-1} = \partial_{12} (z_2-Q)^{-1}$ and the identity
\begin{align*}
R_{12} (z_1-Q_1)^{-1} R_{21} (z_2-Q_2)^{-1} R_{12} ={}& (z_2-Q_2)^{-1} R_{12} (z_1-Q_1)^{-1} \\
&+ \hbar \partial_{12}(z_2-Q_2)^{-1} R_{12} (z_1-Q_1)^{-1} C_{12} R_{12}.
\end{align*}
\end{proof}

\begin{prop}
We have $\mathbf{T}(z)^{-1} = \bar{\mathbf{T}}(z+\gamma\hbar)$.
\end{prop}

\begin{proof}
This follows from the moment map equation
\begin{align*}
AC &= QL-LQ+\gamma\hbar L \quad \text{or} \quad \bar C \bar A = L^{-1}Q-QL^{-1}+\gamma\hbar L^{-1}
\end{align*}
using $(z-Q)^{-1} Q = z (z-Q)^{-1} - 1$:
\begin{align*}
(z-Q)^{-1} AC (z+\gamma\hbar-Q)^{-1}
&= (z-Q)^{-1} L - L (z+\gamma\hbar-Q)^{-1}, \\
(z+\gamma\hbar-Q)^{-1} \bar C \bar A (z-Q)^{-1}
&= L^{-1} (z-Q)^{-1} - (z+\gamma\hbar-Q)^{-1} L^{-1},
\end{align*}
hence
\begin{align*}
\sum_{\rho=1}^\ell (\delta^{\mu\rho} + \bar A^\mu (z-Q)^{-1} A^\rho) (\delta^{\rho\beta} - C^\rho (z+\gamma\hbar-Q)^{-1} \bar C^\beta) - \delta^{\mu\beta}
&= 0, \\
\sum_{\rho=1}^\ell (\delta^{\alpha\rho} - C^\alpha (z+\gamma\hbar-Q)^{-1} \bar C^\rho) (\delta^{\rho\nu} + \bar A^\rho (z-Q)^{-1} A^\nu) - \delta^{\alpha\nu}
&= 0,
\end{align*}
which means $\mathbf{T}(z)^{-1} = \bar{\mathbf{T}}(z+\gamma\hbar)$.
\end{proof}

\begin{remark}
In particular, we notice that the poles of $\bar{\mathbf{T}}(z+\gamma\hbar)$ must be the zeros of $\mathbf{T}(z)$ and vice versa. To compute the position of the poles, we must be careful about the ordering of the $q_i$ inside $\bar{\mathbf{T}}(z)$, since $C$ and $\bar C$ do not commute with $q_i$. We follow the convention that the $q_i$ must be ordered to the left of $C$. In this case, the poles of $\bar{\mathbf{T}}(z)$ arise at $z=q_i-\hbar$, where the shift by $\hbar$ comes about by reordering. It follows that
\begin{equation}
\mathbf{T}(q_i-(\gamma+1)\hbar) = 0.
\end{equation}
\end{remark}

\begin{theorem}
The generator matrices $\mathbf{T}(z)$ and $\bar{\mathbf{T}}(z)$ satisfy the relations
\begin{align}
\hat R^{12}(z_1-z_2) \mathbf{T}^1(z_1) \mathbf{T}^2(z_2) &= \mathbf{T}^2(z_2) \mathbf{T}^1(z_1) \hat R^{12}(z_1-z_2) \\
\bar{\mathbf{T}}^1(z_1) \bar{\mathbf{T}}^2(z_2) \hat R^{12}(z_1-z_2) &= \hat R^{12}(z_1-z_2) \bar{\mathbf{T}}^2(z_2) \bar{\mathbf{T}}^1(z_1) \\
\bar{\mathbf{T}}^1(z_1) \hat R^{21}(z_2-z_1+\gamma\hbar) \mathbf{T}^2(z_2) &= \mathbf{T}^2(z_2) \hat R^{21}(z_2-z_1+\gamma\hbar) \bar{\mathbf{T}}^1(z_1).
\end{align}
with $\bar{\mathbf{T}}(z+\gamma\hbar) = \mathbf{T}(z)^{-1}$ and $\hat R^{12}(z) = \frac{z}{z+\hbar} + \frac{\hbar}{z+\hbar} P^{12}$.
\end{theorem}

Having established the loop algebra and Yangian relations, we my now determine their cross relation via the commutator between the Yangian and the quantum current:

\begin{prop}
The Yangian acts on the quantum current according to:
\begin{equation}\label{eq:crossRel}
\begin{aligned}
[\mathbf{S}^{\alpha\beta}(z_1),\mathbf{T}^{\mu\nu}(z_2)]
={}& \hbar \mathbf{T}^{\mu\beta}(z_2) \partial_{12} \mathbf{S}^{\alpha\nu}(z_1) + \hbar \delta^{\alpha\nu} \sum_{\rho=1}^\ell \mathbf{T}^{\mu\rho}(z_2) \partial_{z_2+\gamma\hbar,z_1} \mathbf{S}^{\rho\beta}(z_1) \\
&+ \frac{\hbar^2 \delta^{\alpha\beta}}{z_1-z_2} \sum_{\rho=1}^\ell \mathbf{T}^{\mu\rho}(z_2) \big( \partial_{z_2+\gamma\hbar,z_1} \mathbf{S}^{\rho\nu}(z_1) - \partial_{z_2+\gamma\hbar,z_2} \mathbf{S}^{\rho\nu}(z_2) \big).
\end{aligned}
\end{equation}
\end{prop}

\begin{proof}
We calculate
\begin{align*}
\mathbf{S}^{\alpha\beta}(z_1) \mathbf{T}^{\mu\nu}(z_2)
={}& \mathbf{S}^{\alpha\beta}(z_1) \delta^{\mu\nu} +\bar A_2^\mu C_1^\alpha \bar R_{21}^{-1} (z_1-Q_1)^{-1} R_{12} (z_2-Q_2)^{-1} R_{12} A_2^\nu A_1^\beta \\
&- \hbar \delta^{\mu\beta} C_1^\alpha (z_1-Q_1)^{-1} \partial_{12} (z_2-Q_2)^{-1} A_2^\nu \\
={}& \mathbf{T}^{\mu\nu}(z_2) \mathbf{S}^{\alpha\beta}(z_1) + \hbar \mathbf{T}^{\mu\beta}(z_2) \partial_{12}\mathbf{S}^{\alpha\nu}(z_1) + \hbar \delta^{\alpha\nu} \sum_{\rho=1}^\ell \mathbf{T}^{\mu\rho}(z_2) \partial_{z_2+\gamma\hbar,z_1} \mathbf{S}^{\rho\beta}(z_1) \\
&+ \frac{\hbar^2 \delta^{\alpha\beta}}{z_1-z_2} \sum_{\rho=1}^\ell \mathbf{T}^{\mu\rho}(z_2) \big( \partial_{z_2+\gamma\hbar,z_1} \mathbf{S}^{\rho\nu}(z_1) - \partial_{z_2+\gamma\hbar,z_2} \mathbf{S}^{\rho\nu}(z_2) \big),
\end{align*}
where we have used 
\begin{align*}
\bar R_{21}^{-1} (z_1&-Q_1)^{-1} R_{21} (z_2-Q_2)^{-1} R_{12} \\
={}& (z_2-Q_2-\hbar Y)^{-1} \bar R_{21}^{-1} \big( (z_1-Q_1)^{-1} + \hbar \partial_{12} (z_1-Q_1)^{-1} C_{12} R_{12} ,
\big),
\end{align*}
and that the moment map equation implies
\begin{equation*}
\bar A^\mu (z_2-Q)^{-1} L (z_1-Q)^{-1} A^\beta = \sum_{\rho=1}^\ell \mathbf{T}^{\mu\rho}(z_2) \frac{\mathbf{S}^{\rho\beta}(z_1)-\mathbf{S}^{\rho\beta}(z_2+\gamma\hbar)}{z_2-z_1+\gamma\hbar},
\end{equation*}
which comes from considering
\begin{equation*}
\bar A^\mu (z_2-Q)^{-1} AC (z_2-Q+\gamma\hbar)^{-1} (z_1-Q)^{-1} A^\beta
\end{equation*}
using $(z_2-Q)^{-1} Q = (z_2-Q)^{-1} z_2 - 1$ as well as the identity $(z_2-Q+\gamma\hbar)^{-1} (z_1-Q)^{-1} = (z_2-z_1+\gamma\hbar)^{-1} \left( (z_1-Q)^{-1} - (z_2-Q+\gamma\hbar)^{-1} \right)$.
\end{proof}

\begin{remark}
This cross relation between $\mathbf{S}(z_1)$ and $\mathbf{T}(z_2)$ determines the cross relations between the Yangian and the loop algebra in $\mathfrak{A}_{N,\ell}$, which we expect to be identified with a truncated affine Yangian. However, carefully matching this cross relation to the relations of affine Yangians is a very non-trivial task, as the known presentations of affine Yangians are highly complicated and presented in very different generators.
\end{remark}

\begin{corollary}
It follows that $\mathbf{J}[1] = \oint dz \mathbf{S}(z)$ commutes with $\mathbf{T}(z)$ according to
\begin{equation}\label{eq:currentYangianRel}
[\mathbf{J}[1]^{\alpha\beta},\mathbf{T}^{\mu\nu}(z)] = -\hbar \mathbf{T}^{\mu\beta}(z) \mathbf{S}^{\alpha\nu}(z) + \hbar \delta^{\alpha\nu} \sum_{\rho=1}^\ell \mathbf{T}^{\mu\rho}(z) \mathbf{S}^{\rho\beta}(z+\gamma\hbar) + \hbar^2 \delta^{\alpha\beta} \sum_{\rho=1}^\ell \mathbf{T}^{\mu\rho}(z) \Delta \mathbf{S}^{\rho\nu}(z),
\end{equation}
where we have introduced the finite difference $\Delta f(z) \coloneq \frac{f(z+\gamma\hbar)-f(z)}{\gamma\hbar}$.
\end{corollary}

\begin{remark}
Let us examine the pole structure of \eqref{eq:currentYangianRel}, where we again have to be careful to order all the $q_i$ to the left. Indeed, we see that the right-hand side has simple poles at $q_i-(\gamma+1)\hbar$ that do not appear on the left-hand side. The vanishing of the residue yields the following additional relation:
\begin{equation}
0 = \gamma\hbar^2 \delta^{\alpha\nu} \sum_{\rho=1}^\ell \mathbf{T}^{\mu\rho}(q_i-(\gamma+1)\hbar) \mathbf{S}_i^{\rho\beta} + \hbar^2\delta^{\alpha\beta} \sum_{\rho=1}^\ell \mathbf{T}^{\mu\rho}(q_i-(\gamma+1)\hbar) \mathbf{S}_i^{\rho\nu},
\end{equation}
but this is trivially satisfied, since we already saw that $\mathbf{T}(q_i-(\gamma+1)\hbar) = 0$. Besides the poles at $q_i-(\gamma+1)\hbar$, we find poles at $z=q_i$ as well as $z=q_i-\hbar$, but these are matched with poles on the left-hand side.
\end{remark}

At this point we can use \eqref{eq:currentYangianRel} to check whether the action of $\mathfrak{A}_{1,\ell}$ coincides with the action of the loop Yangian in the simplest non-trivial representation---the vector representation $\C^\ell[q]$. Concretely, the action of the Yangian and loop algebra on $\C^\ell[q]$ are
\begin{equation}
    \mathbf{T}^{\alpha\beta}(z) \mapsto \delta_{\alpha\beta} - \frac{\hbar}{z-q} e_{\alpha\beta}, \qquad \mathbf{J}[1]^{\alpha\beta} \mapsto e_{\alpha\beta} e^{-\hbar\partial}.
\end{equation}
Due to our specification $N=1$, we have $\mathbf{J}[1] = \mathbf{S}_1$ and $\mathbf{S}(z) = \frac{1}{z-q_1+\hbar} \mathbf{S}_1$, and indeed we find that \eqref{eq:currentYangianRel} is satisfied for $\gamma = -\ell-1$. More generally, we conjecture that $\mathfrak{A}_{N,\ell}$ can be represented on the Schur--Weyl bimodule $(\C^\ell)^{\otimes N} \otimes \C[q_1,\dots,q_N]$, which also appears in the quantum spin Ruijsenaars--Schneider models of \cite{klabbers:2024b}.

\section{Calculation of the spectrum} \label{section:spectrum}

We now come to a proposal to calculate eigenvalues and eigenvectors of the lowest Hamiltonian $H \coloneq \operatorname{Tr} \mathbf{J}[1] = \sum_{i=1}^N \operatorname{Tr} \mathbf{S}_i$. While it is a very ambitious goal to find eigenfunctions for the rational potential even in the spinless case $\ell=1$, requiring heavy methods of functional analysis \cite{hallnas:2014,hallnas:2018,hallnas:2021}, we can hope to at least discover the spectrum and Fock space structure of an interesting subsector of the model with the help of the Yangian. Indeed, from \eqref{eq:currentYangianRel}, we derive
\begin{equation}\label{eq:HamiltonianYangianRel}
\begin{aligned}
[H,\mathbf{T}(z)]
={}& \hbar^2 (\gamma+\ell) \mathbf{T}(z) \Delta \mathbf{S}(z),
\end{aligned}
\end{equation}
which closely resembles the way creation operators should commute with the Hamiltonian. We remark that the Yangian becomes a symmetry of the Hamiltonian at $\gamma = -\ell$, which is the critical level in the language of \cite{reshetikhin:1990}. The appearance of a critical level should be compared with the Gaudin model, where the commutator between the Segal--Sugawara element and the current of the affine Kac--Moody algebra is proportional to $\kappa-\kappa_c$ (where $\kappa$ is the level and $\kappa_c$ is the critical level) \cite{frenkel:2010}. On the other hand, we conclude that in the case of non-critical level $\gamma \neq -\ell$, we can use the generators of the Yangian as creation operators to move between energy levels and build up the spectrum of the Hamiltonian.

To this end, let $\mathbf{E} \coloneq \oint dz \mathbf{T}(z)$, which is the generator of $\mathfrak{gl}_\ell$-symmetry and commutes with $\mathbf{S}(z)$ according to
\begin{equation}
[\mathbf{S}^{\alpha\beta}(z),\mathbf{E}^{\mu\nu}] = -\hbar (\delta^{\mu\beta} \mathbf{S}^{\alpha\nu}(z) - \delta^{\alpha\nu} \mathbf{S}^{\mu\beta}(z)).
\end{equation}
While it is at the moment unclear to us how to formulate a highest weight condition for both the Yangian and the Hamiltonian $H$, we may consider a vacuum vector $|0\rangle$ satisfying
\begin{equation}
\mathbf{S}^{\alpha\beta}(z) |0\rangle = \delta^{\alpha\beta} \sum_{i=1}^N \frac{1}{z-q_i+\hbar} \prod_{j(\neq i)} \frac{q_{ij}-\hbar}{q_{ij}} |0\rangle, \quad \mathbf{E}^{\alpha\alpha} |0\rangle = \delta^{\alpha\ell} N |0\rangle, \quad \mathbf{E}^{\alpha\beta}|0\rangle = 0, \quad \alpha > \beta.
\end{equation}
It is easily checked that these conditions are compatible with the commutation relations of $\mathbf{S}(z)$ and $\mathbf{E}$ when $\gamma = 1$, so we will set $\gamma = 1$ in the following. It follows that the vacuum energy is given by
\begin{equation}
H |0\rangle = \ell \sum_{i=1}^N \prod_{j(\neq i)} \frac{q_{ij}-\hbar}{q_{ij}} |0\rangle = N \ell |0\rangle.
\end{equation}
We now consider the creation operator $\bar{\mathbf{A}}^\mu(w) \coloneq \mathbf{T}^{\mu\ell}(w)$ for $\mu = 1,\dots,\ell$, which is equal to $\bar A^\mu (w-Q)^{-1} e$ under the gauge condition $A_i^\ell=1$, and build up states of the form
\begin{equation}
\bar{\mathbf{A}}^{\mu_1}(w_1) \cdots \bar{\mathbf{A}}^{\mu_k}(w_k) |0\rangle,
\end{equation}
which are subject to the exchange relation
\begin{equation}
\begin{aligned}
\bar{\mathbf{A}}^1(w_1) &\cdots \bar{\mathbf{A}}^a(w_a) \bar{\mathbf{A}}^{a+1}(w_{a+1}) \cdots \bar{\mathbf{A}}^k(w_k) |0\rangle \\
={}& \hat R^{a,a+1}(w_a-w_{a+1}) \bar{\mathbf{A}}^1(w_1) \cdots \bar{\mathbf{A}}^{a+1}(w_{a+1}) \bar{\mathbf{A}}^a(w_a) \cdots \bar{\mathbf{A}}^k(w_k) |0\rangle,
\end{aligned}
\end{equation}
where we have used auxiliary space notation in spin space. Acting with $\mathbf{S}^{\alpha\alpha}(z)$ for $\alpha = 1,\dots,\ell-1$, we obtain
\begin{equation}
\begin{aligned}
\mathbf{S}^{\alpha\alpha}(z) & \bar{\mathbf{A}}^{\mu_1}(w_1) \cdots \bar{\mathbf{A}}^{\mu_k}(w_k) |0\rangle \\
={}& \frac{1}{2} \sum_{i=1}^N \frac{1}{z-q_i+\hbar} \bigg( 1 + \prod_{a=1}^k (1-2\hbar V(w_a-q_i+\hbar)) \bigg) \prod_{j(\neq i)} \frac{q_{ij}-\hbar}{q_{ij}}\bar{\mathbf{A}}^{\mu_1}(w_1) \cdots \bar{\mathbf{A}}^{\mu_k}(w_k) |0\rangle
\end{aligned}
\end{equation}
as well as
\begin{equation}
\begin{aligned}
\mathbf{S}^{\ell\ell}(z) & \bar{\mathbf{A}}^{\mu_1}(w_1) \cdots \bar{\mathbf{A}}^{\mu_k}(w_k) |0\rangle \\
={}& \sum_{i=1}^N \frac{1}{z-q_i+\hbar} \prod_{a=1}^k (1-2\hbar V(w_a-q_i+\hbar)) \prod_{j(\neq i)} \frac{q_{ij}-\hbar}{q_{ij}} \bar{\mathbf{A}}^{\mu_1}(w_1) \cdots \bar{\mathbf{A}}^{\mu_k}(w_k) |0\rangle,
\end{aligned}
\end{equation}
where we have introduced the potential $V(z) \coloneq \frac{1}{z}-\frac{1}{z+\hbar}$. This implies
\begin{equation}
\begin{aligned}
H \bar{\mathbf{A}}^{\mu_1}&(w_1) \cdots \bar{\mathbf{A}}^{\mu_k}(w_k) |0\rangle \\
={}& \bigg( \frac{N(\ell-1)}{2} + \frac{\ell+1}{2} \sum_{i=1}^N \prod_{a=1}^k (1-2\hbar V(w_a-q_i+\hbar)) \prod_{j(\neq i)} \frac{q_{ij}-\hbar}{q_{ij}} \bigg) \bar{\mathbf{A}}^{\mu_1}(w_1) \cdots \bar{\mathbf{A}}^{\mu_k}(w_k) |0\rangle.
\end{aligned}
\end{equation}
Importantly, this is not an eigenvalue equation since the term in front depends on $q_1,\dots,q_N$, which do not act as a scalar on the vacuum. However, we can define the rescaled and shifted Hamiltonian $\tilde H \coloneq \frac{2}{\ell+1} H - \frac{\ell-1}{\ell+1} N$ and multiply the state with an arbitrary function $\psi(q_1,\dots,q_N)$ giving
\begin{equation}
\begin{aligned}
\tilde H \psi&(q_1,\dots,q_N) \bar{\mathbf{A}}^{\mu_1}(w_1) \cdots \bar{\mathbf{A}}^{\mu_k}(w_k) |0\rangle \\
={}& \sum_{i=1}^N \psi(q_1,\dots,q_i-\hbar,\dots,q_N) \prod_{a=1}^k (1-2\hbar V(w_a-q_i+\hbar)) \prod_{j(\neq i)} \frac{q_{ij}-\hbar}{q_{ij}} \bar{\mathbf{A}}^{\mu_1}(w_1) \cdots \bar{\mathbf{A}}^{\mu_k}(w_k) |0\rangle.
\end{aligned}
\end{equation}
This becomes an eigenvalue equation when $\psi$ fulfills the difference equation
\begin{equation} \label{eq:diffeq}
\sum_{i=1}^N \prod_{a=1}^k (1-2\hbar V(w_a-q_i+\hbar)) \prod_{j(\neq i)} \frac{q_{ij}-\hbar}{q_{ij}} e^{-\hbar\partial_i} \psi = \lambda \psi.
\end{equation}
It should be remarked that the energy eigenvalues of the eigenstates determined by solutions of \eqref{eq:diffeq} do not depend on the spin configuration given by $\mu_1,\dots,\mu_k$. This is a phenomenon observed more generally in spin-type integrable many-body systems. For example, the spectrum of the rational spin Calogero model considered in \cite{bourgine:2024} does not depend on the underlying spin configurations either. Instead, the introduction of spins enters into the degeneracies of the energy spectrum and modifies the exchange statistics.

We note that for $k=0$ the difference operator in \eqref{eq:diffeq} reproduces exactly the spinless quantum Ruijsenaars--Schneider Hamiltonian with shift parameter $\gamma = 1$. Its simplest entire solution is given by plane waves multiplied with an attractive and symmetric Vandermonde-type factor:
\begin{equation}
\psi(q_1,\dots,q_N) \coloneq \exp \bigg( \sum_{i=1}^N \frac{p_i q_i}{\hbar} \bigg) \Delta(q_1,\dots,q_N), \quad \Delta(q_1,\dots,q_N) \coloneq \prod_{i < j} \operatorname{sinc} \bigg( 2\pi \frac{q_i-q_j}{\hbar} \bigg)
\end{equation}
The corresponding eigenvalue is $\lambda = \sum_{i=1}^N e^{-p_i}$ and its $\mathfrak{gl}_\ell$-degeneracy is $\operatorname{Sym}^N \C^\ell$ coming from the action of $\mathbf{E}$ on the vacuum. We also know that $\mathbf{J}[1]^{\alpha\alpha}$ for $\alpha = 1,\dots,\ell$ acts as a symmetry of the Hamiltonian, but we find
\begin{equation}
\mathbf{J}[1]^{\alpha\alpha} \psi |0\rangle = \sum_{i=1}^N \prod_{j(\neq i)} \frac{q_{ij}-\hbar}{q_{ij}} e^{-\hbar\partial_i} \psi |0\rangle = \lambda \psi |0\rangle,
\end{equation}
so this does not produce more states. We conclude that the $L(\mathfrak{gl}_\ell)$-degeneracy of these eigenstates is just the evaluation module $\operatorname{Sym}^N \C^\ell$. This should be compared with subsection 3.1.4 of \cite{lamers:2022}, where it is found that the ground state of their model has degeneracy $\operatorname{Sym}^N \C^\ell$.

In the case $k=1$, we find that the mirabolic subalgebra $\mathfrak{m_\ell} = \langle e_{\alpha\beta} \mid \alpha < \ell \rangle \subseteq \mathfrak{gl}_\ell$ acts as
\begin{equation}
\mathbf{E}^{\alpha\beta} \psi \bar{\mathbf{A}}^\mu(w) |0\rangle
= -\hbar \delta^{\mu\beta} \psi \bar{\mathbf{A}}^{\alpha}(w) |0\rangle + \psi \bar{\mathbf{A}}^\mu(w) \mathbf{E}^{\alpha\beta} |0\rangle.
\end{equation}
However, the action of the full $\mathfrak{gl}_\ell$ contains extra terms involving Yangian generators that cannot be simplified without imposing further relations on this representation, which is beyond the scope of the current paper. This and a detailed description of the spectrum for higher $k$ will be subject of future work.

\section{Conclusion}

In this work we have constructed a quantization of the classical rational spin Ruijsenaars–Schneider system defined in \cite{krichever:1995} by performing a quantum Hamiltonian reduction starting with the quantum cotangent bundle $\mathcal{O}_\hbar(T^* \mathrm{GL}_N)$ \cite{arutyunov:1996} tensored with canonical oscillators $\mathcal{O}_\hbar(T^* \C^{N \times \ell})$, obtaining a family of algebras $\mathfrak{A}_{N,\ell}$ that may be interpreted as quantized mixed quiver varieties for the framed Jordan quiver. This is in direct analogy with the Hamiltonian reduction in the classical case \cite{arutyunov:1998}. We provide a functional operator presentation of $\mathcal{O}_\hbar(T^* \mathrm{GL}_N)$ on the Lagrange basis and provide two methods to determine the quantum Hamiltonian reduction, first by step-wise reduction with respect to the mirabolic subgroup and the torus, and second by fully invariant generating functions $\mathbf{S}[n](z)$ that can be rewritten in terms of $\mathbf{S}[1](z)$, whose residues satisfy the quantum current relation \eqref{eq:SSrelation}. Within the algebra $\mathfrak{A}_{N,\ell}$ we identified the commuting Hamiltonians via the center of a loop algebra $L(\mathfrak{gl}_\ell)$ and exhibited the Yangian $Y(\mathfrak{gl}_\ell)$ as a subalgebra. We computed the spectrum of the lowest Hamiltonian and found that the energy eigenvalues do not directly depend on the spin configuration, cf.\ \cite{bourgine:2024}, but only indirectly via exchange statistics. Finally, we obtained a difference equation for eigenstates of the lowest Hamiltonian, solving for its ground states, which match the findings of \cite{lamers:2022}.

It remains to extend the representations we constructed to the whole algebra $\mathfrak{A}_{N,\ell}$. In the spirit of the construction of quantum spin Ruijsenaars--Schneider models that appeared in \cite{lamers:2022}, we conjecture that $\mathfrak{A}_{N,\ell}$ can be represented on the Schur--Weyl bimodule $(\C^\ell)^{\otimes N} \otimes \C[q_1,\dots,q_N]$, which would also open the door to a proof of an identification with a truncated affine Yangian, which has canonical actions on the Schur--Weyl modules that are jointly faithful. This will be the subject of future work. More far-fetchedly, we would like to generalize the construction at hand to the most general elliptic potential, cf.\ \cite{matushko:2023,klabbers:2024b,klabbers:2024}, and ideally even to the elusive double elliptic (`DELL') models \cite{mironov:2024}.

\vspace{0.3\baselineskip}

\noindent\textbf{Acknowledgments.} Our deepest gratitude goes to J. Lamers for many extensive and fruitful discussions as well as useful comments on an early draft. Our thanks also go to A. Shapiro and J. Teschner for elucidating the connection to moduli spaces in gauge theory and C. Raymond for pointing us to literature on quantum Hamiltonian reduction as well as L. Fehér for discussions on Hamiltonian reduction.

\vspace{0.3\baselineskip}

\noindent\textbf{Funding.} GA acknowledges support by the DFG under Germany's Excellence Strategy -- EXC 2121 ``Quantum Universe'' -- 390833306. GA and LH acknowledge support by the DFG -- SFB 1624 -- ``Higher structures, moduli spaces and integrability'' -- 506632645.

\vspace{0.3\baselineskip}

\noindent\textbf{Competing interests.} The authors have no relevant financial or non-financial interests to disclose.

\vspace{0.3\baselineskip}

\noindent\textbf{Data availability.} The authors declare that the data supporting the findings of this study are available within the paper.

\pagebreak

\appendix

\section{Review of quantum Hamiltonian reduction}\label{appendix:reduction}

The technique of quantum Hamiltonian reduction is a simplified version of BRST quantization that takes as an input an algebra $\mathcal{A}$ of quantum observables together with a quantum moment map $\mu \colon \mathcal{H} \to \mathcal{A}$, which is a homomorphism from a Hopf algebra $\mathcal{H}$ to $\mathcal{A}$, as well as an ideal $\mathcal{I}$ in $\mathcal{H}$. The quantum moment map encodes an action of $\mathcal{H}$ on $\mathcal{A}$ via
\begin{equation}
\operatorname{ad}_h a \coloneq \mu(h_{(1)}) a \mu(S(h_{(2)})), \quad h \in \mathcal{H}, \quad a \in \mathcal{A},
\end{equation}
where we have used Sweedler notation for the coproduct of $\mathcal{H}$. In the case where $\mathcal{H} = U(\mathfrak{g})$ is the universal enveloping algebra of a Lie algebra, the formula reduces to $\operatorname{ad}_h a = [h,a]$. Let $\mathcal{A}^\mathcal{H} \coloneq \{ a \in \mathcal{A} \mid \operatorname{ad}_h a = \epsilon(h) a \}$ be the algebra of invariants under the action by $\mathcal{H}$, where $\epsilon$ is the counit of $\mathcal{H}$. Further, let $\mathcal{I} \coloneq \mathcal{A} \mu(\mathcal{I})$. Then the quantum Hamiltonian reduction of $\mathcal{A}$ by $\mathcal{H}$ may be defined by
\begin{equation}
\mathcal{A} \sslash \mathcal{H} \coloneq (\mathcal{A}/\mathcal{I})^\mathcal{H},
\end{equation}
equipped with the induced multiplication from $\mathcal{A}$ which is well-defined after taking invariants. In the case where $\mathcal{H} = U(\mathfrak{g})$, taking invariants means that we are restricting to the commutant of the image of the quantum moment map.

The analogy with classical Hamiltonian reduction of a symplectic manifold $M$ by a Hamiltonian action of a Lie group $G$ with Lie algebra $\mathfrak{g}$ emerges when dualizing the classical moment map $\mu \colon M \to \mathfrak{g}^*$ to obtain a morphism of Lie algebras
\begin{equation}
\mathfrak{g} \to C^\infty(M), \quad x \mapsto (p \in M \mapsto \langle \mu(p),x \rangle),
\end{equation}
which maps elements of the symmetry algebra $\mathfrak{g}$ to the algebra of observables $C^\infty(M)$, which is exactly the role of a quantum moment map.

A good reference on the topic of quantum Hamiltonian reduction providing the general theory as well as examples is \cite{notes:etingof:2009}. Importantly, it has been employed in the quantization of Hitchin systems \cite{beilinson:1991} and has more recently recently appeared in the context of factorization homology \cite{benzvi:2018}, \cite{benzvi:2018b}, where the algebra of quantum observables of the trigonometric spinless Ruijsenaars--Schneider model, also known as the spherical double affine Hecke algebra, was considered as a quantum Hamiltonian reduction in the context of factorization homology of the punctured torus with coefficients in the quantum group $U_q(\mathfrak{gl}_N)$. The paper \cite{costello:2017} also makes use of quantum Hamiltonian reduction to construct the deformed double current algebra.

\pagebreak

\section*{Notation}\label{section:notation}
\addcontentsline{toc}{section}{Notation}

\begin{tabular}{r|l}
$\hbar$ & Planck constant of the Ruijsenaars--Schneider model \\
$\gamma$ & coupling constant of the Ruijsenaars--Schneider model \\
$N$ & number of particles \\
$\ell$ & number of spin states \\
$i,j,k,l$ & particle index ranging from $1$ to $N$ \\
$\alpha,\beta,\mu,\nu,\rho$ & spin state index ranging from $1$ to $\ell$ \\
$\mathrm{GL}_N$ & general linear group of rank $N$ \\
$\mathrm{T}_N$ & torus of diagonal matrices of rank $N$ \\
$\mathrm{M}_N$ & mirabolic subgroup of matrices of rank $N$ \\
$\mathfrak{gl}_N$ & general linear Lie algebra of rank $N$ \\
$\mathfrak{m}_N$ & mirabolic Lie subalgebra of rank $N$ \\
$U(\mathfrak{g})$ & universal enveloping algebra of a Lie algebra $\mathfrak{g}$ \\
$L(\mathfrak{gl}_\ell)$ & universal enveloping algebra of the loop algebra of $\mathfrak{gl}_\ell$ \\
$Y(\mathfrak{gl}_\ell)$ & Yangian of $\mathfrak{gl}_\ell$ \\
$\mathfrak{A}_{N,\ell}$ & algebra of observables of the rational spin Ruijsenaars--Schneider model \\
$T^* \mathrm{GL}_N$ & cotangent bundle of the general linear group of rank $N$ \\
$\mathcal{O}_\hbar(T^* \mathrm{GL}_N)$ & quantum cotangent bundle of the general linear group of rank $N$ \\
$T^* \C^{N \times \ell}$ & Poisson manifold of canonical oscillators \\
$\mathcal{O}_\hbar(T^* \C^{N \times \ell})$ & algebra of canonical oscillators \\
$a,b$ or $a_i^\alpha, b_j^\beta$ & generators of $\mathcal{O}_\hbar(T^* \mathrm{GL}_N)$ \\
$q_i$ & position of $i$th particle \\
$q_{ij}$ & $q_i-q_j$ \\
$Q$ & diagonal matrix of positions \\
$P$ & diagonal matrix of exponential momenta \\
$T,U$ & matrices of generators of $\mathcal{O}_\hbar(T^* \mathrm{GL}_N)$ \\
$e$ & length $N$ column vector with all entries equal to one \\
$e_i$ & $i$th unit vector of $\C^N$ \\
$e_{ij}$ & $N \times N$ matrix unit for $i$th row and $j$th column \\
$e_\alpha$ & $\alpha$th unit vector of $\C^\ell$ \\
$e_{\alpha\beta}$ & $\ell \times \ell$ matrix unit for $\alpha$th row and $\beta$th column \\
$C_{12}$ & particle permutation operator $\sum_{i,j=1}^N e_{ij} \otimes e_{ji}$ \\
$P^{12}$ & spin permutation operator $\sum_{\alpha,\beta=1}^\ell e_{\alpha\beta} \otimes e_{\beta\alpha}$ \\
$\hat R^{12}(z)$ & Yang's $R$-matrix in various normalizations \\
$\partial_{12}$ or $\partial_{z_1,z_2}$ & Newton divided difference in the variables $z_1,z_2$ \\
$X_{12}$ & $\sum_{i,j=1}^N e_{ij} \otimes e_{jj}$ \\
$Y_{12}$ & $\sum_{i=1}^N e_{ii} \otimes e_{ii}$ \\
$\ell_i(z)$ & Lagrange interpolation polynomial $\prod_{j(\neq i)} \frac{z-q_j}{q_i-q_j}$ \\
$\ell_i^*(z)$ & dual Lagrange interpolation polynomial $\prod_{j(\neq i)} (z-q_j)$
\end{tabular}

\begin{tabular}{r|l}
$\delta(z,z')$ & $\delta$-polynomial $\sum_{i=1}^N \ell_i(z) \ell_i^*(z')$ \\
$\chi(z)$ & characteristic polynomial $\prod_{i=1}^N (z-q_i)$ \\
$\F$ & field of rational functions $\C(q_1,\dots,q_N)$ \\
$\F[z]_N$ & ring of polynomials with coefficients in $\F$ and degree $<N$ \\
$\langle-,-\rangle$ & inner product on $\F[z]_N$ $\langle f,g \rangle = \oint \frac{dz}{\chi(z)} f(z) g(z)$ \\
$s_{12}$ & permutation of $z_1$ and $z_2$ \\
$\partial_{12}$ & divided difference $(z_1-z_2)^{-1} (1-s_{12})$ \\
$\partial_{12}^*$ & adjoint of divided difference with respect to $\langle-,-\rangle$ \\
$\mathbf{R}_{12}^*$ & $1+\hbar \partial_{12}$ \\
$\mathbf{R}_{12}$ & $1+\hbar \partial_{12}^*$ \\
$R_{12}$ & $1+\sum_{i \neq j} \frac{\hbar}{q_{ij}} f_{ij} \otimes f_{ji}$ \\
$\bar R_{12}$ & $1+\sum_{i \neq j} \frac{\hbar}{q_{ij}-\hbar} f_{ij} \otimes e_{jj}$ \\
$A$ & dressed oscillator $T^{-1} a$ \\
$B$ & dressed oscillator $bU$ \\
$C$ & dressed oscillator $bUP$ \\
$\bar A$ & dressed oscillator $CL^{-1}$ \\
$\bar C$ & dressed oscillator $L^{-1}A$ \\
$W$ & Frobenius matrix $T^{-1}U$ \\
$L$ & Lax matrix $T^{-1}UP$ \\
$A(z)$ & polynomial dressed oscillator $e^t \ell(z) A$ \\
$C(z)$ & polynomial dressed oscillator $C \ell^*(z) e$ \\
$L(z,z')$ & polynomial Lax matrix $e^t \ell(z) L \ell^*(z') e$ \\
$\mathbf{S}[n](z)$ & generating function $C L^{n-1} (z-Q)^{-1} A$ \\
$\mathbf{S}(z)$ & quantum current $C (z-Q)^{-1} A$ \\
$\mathbf{J}[n]$ & loop algebra generator $C L^{n-1} A$ \\
$\mathbf{T}(z)$ & Yangian generator $1+\bar A (z-Q)^{-1} A$ \\
$\bar{\mathbf{T}}(z)$ & dual Yangian generator $1-C (z-Q)^{-1} \bar C$ \\
$\Delta$ & finite difference $\Delta f(z) = \frac{f(z+\gamma\hbar)-f(z)}{\gamma\hbar}$ \\
$V(z)$ & rational potential $\frac{1}{z} - \frac{1}{z+\hbar}$ \\
$H$ & lowest Hamiltonian $\operatorname{Tr} \mathbf{J}[1]$ \\
$\tilde H$ & rescaled Hamiltonian $\frac{2}{\ell+1} - \frac{\ell-1}{\ell+1} N$ \\
$\mathbf{E}$ & $\mathfrak{gl}_\ell$-generators $\oint dz \mathbf{T}(z)$ \\
$\bar{\mathbf{A}}^\mu(z)$ & creation operator $\mathbf{T}^{\mu\ell}(z)$ \\
$\oint dz$ & negative residue at infinity
\end{tabular}

\pagebreak

\bibliographystyle{plainurl}
\bibliography{bibliography}
\end{document}